\documentclass[showpacs,twocolumn,pra,superscriptaddress,notitlepage]{revtex4-2}
\usepackage{qcircuit}
\usepackage{graphics}
\usepackage{amsmath,amssymb,amsthm,mathrsfs,amsfonts,dsfont}
\usepackage{subfigure, epsfig}
\usepackage{braket}
\usepackage{bm}
\usepackage{enumerate}
\usepackage{algorithm}
\usepackage{physics}

\usepackage{pgfplots}
\pgfplotsset{width=7cm,compat=1.9}

\usepgfplotslibrary{external}

\usepackage{algpseudocode}
\usepackage{color}
\usepackage{comment}
\usepackage[colorlinks = true]{hyperref}

\newtheorem{lemma}{Lemma}

\newtheorem{proposition}{Proposition}

\newcommand{\renyi}{R$\mathrm{\acute{e}}$nyi }
\newcommand{\mc}{\mathcal}
\newcommand{\mb}{\mathbf}
\newcommand{\mbb}{\mathbb}

\newcommand{\sub}[2]{ \mathop{#1}\limits_{#2}}
\newcommand{\pr}{\mathrm{Pr}}

\newcommand{\comments}[1]{}

\begin{document}
\title{Characterizing correlation within multipartite quantum systems via local randomized measurements}

\author{Zhenhuan Liu}
\email{liu-zh20@mails.tsinghua.edu.cn}
\affiliation{Center for Quantum Information, Institute for Interdisciplinary Information Sciences, Tsinghua University, Beijing 100084, China}
\author{Pei Zeng}
\email{peizeng.phy@gmail.com}
\affiliation{CAS Center for Excellence in Quantum Information and Quantum Physics, University of Science and Technology of China, Hefei, Anhui 230026, China}
\affiliation{Center for Quantum Information, Institute for Interdisciplinary Information Sciences, Tsinghua University, Beijing 100084, China}
\author{You Zhou}
\email{zyqphy@gmail.com}
\affiliation{Nanyang Quantum Hub, School of Physical and Mathematical Sciences, Nanyang Technological University, Singapore 637371}
\affiliation{Centre for Quantum Technologies, National University of Singapore, 3 Science Drive 2, 117543 Singapore}
\author{Mile Gu}
\email{mgu@quantumcomplexity.org}
\affiliation{Nanyang Quantum Hub, School of Physical and Mathematical Sciences, Nanyang Technological University, Singapore 637371}
\affiliation{Centre for Quantum Technologies, National University of Singapore, 3 Science Drive 2, 117543 Singapore}

\begin{abstract}
Given a quantum system on many qubits split into a few different parties, how many total correlations are there between these parties? Such a quantity, aimed to measure the deviation of the global quantum state from an uncorrelated state with the same local statistics, plays an important role in understanding multipartite correlations within complex networks of quantum states. Yet, the experimental access of this quantity remains challenging as it tends to be non-linear, and hence often requires tomography which becomes quickly intractable as dimensions of relevant quantum systems scale. Here, we introduce a much more experimentally accessible quantifier of total correlations, which can be estimated using only single-qubit measurements. It requires far fewer measurements than state tomography, and obviates the need to coherently interfere multiple copies of a given state. Thus we provide a tool for proving multipartite correlations that can be applied to near-term quantum devices.
\end{abstract}
\maketitle

\section{Introduction}\label{sec:intro}


The preparation of highly correlated quantum states across many qubits is essential for advanced quantum information processing \cite{preskill2018quantum,horodecki2009quantum,modi2012classical}. Yet, in the noisy intermediate-scale quantum (NISQ) era, techniques for doing so are not necessarily reliable. Consequently, there is surging interest in quantum benchmarking \cite{Eisert2020certification,Kliesch2021Certification} --- identifying efficient means of verifying what a quantum computer is doing compared to what it is meant to do. Of these, an analysis of how many correlations exist across many qubits faces significant challenges owing to the exponentially growing size of Hilbert space. This is especially true when there is no prior information regarding how the state is prepared. 

\begin{figure}[htbp]
    \centering
    \includegraphics[scale=0.28]{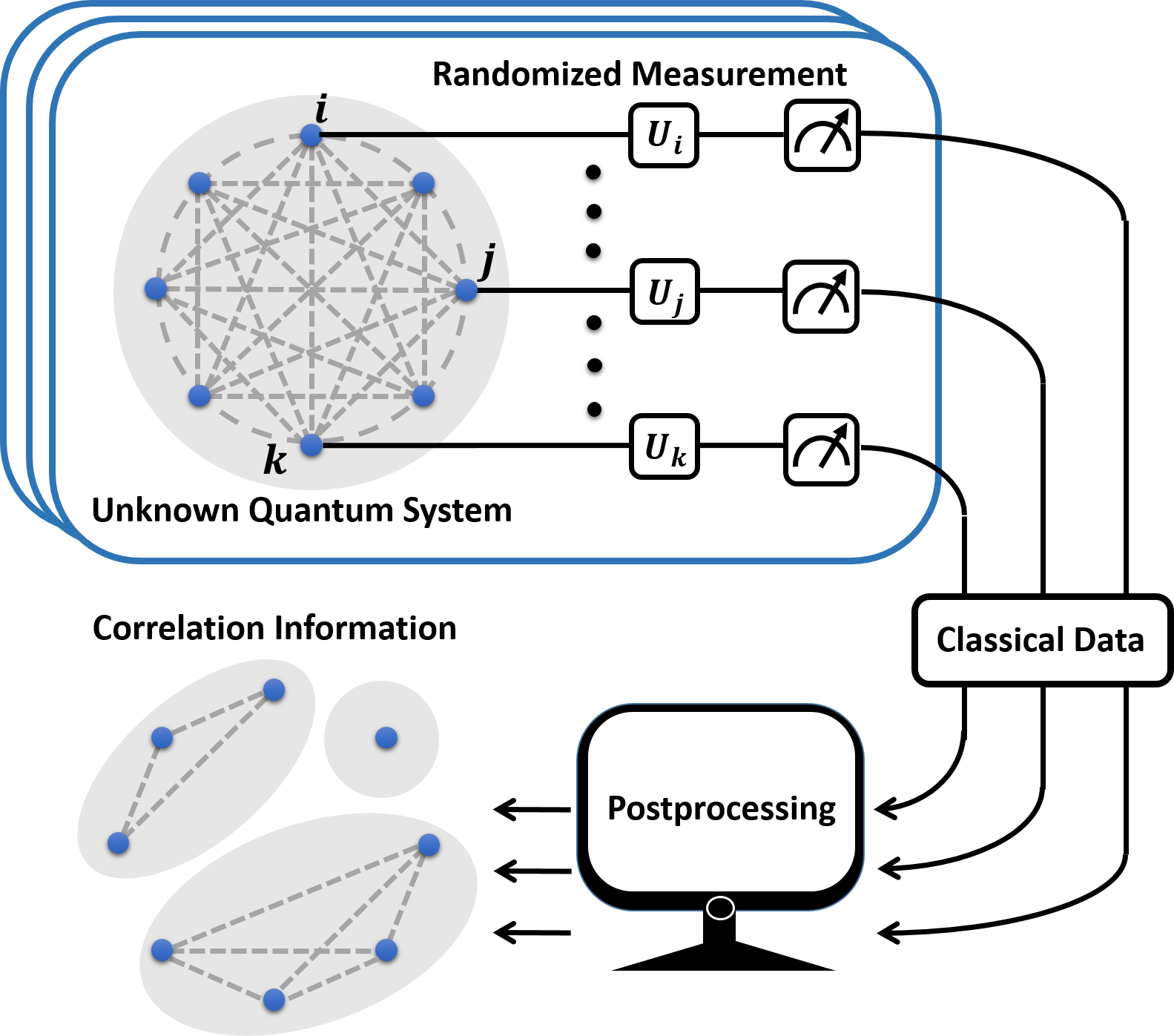}
    \caption{Illustration of the main idea. We present a protocol to measure the total correlation of an unknown multipartite quantum state in any partition. We perform qubitwise local randomized measurements on sequentially prepared states and thus obtain the classical measurement data. Using tailored data postprocessing strategies, the correlation information with respect to any chosen partition can be extracted.}
    \label{fig:Structure}
\end{figure}

One key amount of common interest is the total correlation within a multipartite quantum system \cite{modi2010unified,modi2012classical}. Consider a joint quantum system consisting of $k$ subsystems $\{g_1,g_2, \ldots, g_k\}$, where each subsystem $g_i$ has local statistics specified by respective density operator $\rho_i$. The joint system would be said to have no correlation if the joint state obeys $\rho = \bigotimes_{i=1}^k\rho_i$, such that the global statistics is simply the product of its marginals. A state then possesses correlations if there exists deviation from this tensor product. A common measure of such deviation is the relative entropy, resulting in the quantifier $S(\rho\|\bigotimes_{i=1}^k\rho_i)=\sum_iS(\rho_i)-S(\rho)$, where $S(\cdot)$ represents the von Neumann entropy. The quantity has found applications in quantum thermodynamics \cite{Goold2016thermoreview} and many-body physics \cite{de2018genuine, Goold2015MBL}, and the characterization of genuine multipartite correlation \cite{Bennett_2011,giorgi2011genuine,Girolami_2017}. Nevertheless, the quantity remains difficult to access in practical experiments due to its inherent nonlinearity. Most approaches would require either interacting multiple copies of $\rho$ or state tomography, tasks that can be prohibitive if $\rho$ already presents the most challenging state one can synthesize on NISQ devices.

Here, we propose a quantifier of the total correlation --- \emph{correlation overlap} --- whose technological accessibility is much closer to the synthesis of $\rho$ itself. In particular, our protocol only requires repeated synthesis of the same $\rho$, together with local (qubitwise) random unitary evolution and computational measurements (see Fig.~\ref{fig:Structure}). Specifically, we show that the \emph{correlation overlap} can be obtained by postprocessing the measurement data and the amount of data required is much less than the traditional quantum tomography. Meanwhile, the quantity itself maintains its operational meaning as a quantifier of total correlations, and can also be immediately adapted to measure how close candidate systems are to the maximally entangled.

\section{Definition}\label{sec:def}
Recall that if a $k$-partite quantum state $\rho$ is uncorrelated, it can be written as $\bigotimes_{i=1}^k\rho_i$ with $\rho_i=\tr_{\bar{i}}(\rho)$ being the reduced density matrix of the $i$-th subsystem. Normally the relative entropy is adopted to quantify the distance between them \cite{modi2010unified,modi2012classical}. The von Neumann entropy involved can be in principle acquired by the quantum state tomography, which is already challenging for systems with more than ten qubits. Thus, in order to make the measurement protocol scalable, one needs to avoid state tomography \cite{mandal2020coherence}. Alternative entropy functions such as \renyi entropy can be obtained by measuring the purity of the state 
\cite{Islam2015Measuring,Kaufmanen2016tanglement,Brydges2019Probing}. However, \renyi entropy can violate the subadditivity \cite{horodecki2009quantum,linden2013structure}, which makes it nonideal for quantifying total correlation. Alternative approaches include uses of witnesses \cite{GUHNE2009detection,Friis2019Reviews} to detect the presence of certain correlations \cite{Huber2013Structure,Shahandeh2014Structural,Lu2018Structure,zhou2019detecting}. These witnesses are typically tailored for specific classes of states (e.g. \cite{toth2005detecting,zhou2019detecting}) and are generally ineffective when applied to states without the preparation information \cite{Zhu2010Minimal,Dai2014Witness,You2020coherent}.

Here we quantify the total correlation based on the fidelity between a given state $\rho$ and its marginals as follows:
\begin{equation}\label{eq:correlation def}
\begin{aligned}
&C(\rho)=-\log\mathcal{F}\left(\rho,\bigotimes_{i=1}^k\rho_i\right),
\end{aligned}
\end{equation}
with the fidelity \cite{Liang_2019} being

\begin{equation}\label{eq:F2correlation2}
\begin{aligned}
\mathcal{F}\left(\rho,\bigotimes_{i=1}^k\rho_i\right)=\frac{\tr\left(\rho \bigotimes_{i=1}^k\rho_i\right)}{\sqrt{\tr(\rho^2)\left[\prod_{i=1}^k\tr(\rho_i^2)\right]}}.
\end{aligned}
\end{equation}

Notice that $k$ is not necessarily the number of qubits, but the number of subsystems under some partition. In Appendix \ref{app:discussion property}, we show that such total correlation measure satisfies certain key properties, such as faithfulness, no change under local unitary transformation, and additivity under tensor product. By taking the minimization on all possible bipartitions, one can also generalize it to quantify genuine multipartite correlation. We remark that other fidelity measures \cite{Liang_2019} could also be adapted to define the total correlation where the measurability is the main concern.


The denominator of Eq.~\eqref{eq:F2correlation2} is composed of a few purity terms, and there already exist effective methods to measure them \cite{Ekert_2002,Brydges2019Probing}. The main contribution of this work is that we develop a protocol to effectively measure the numerator
\begin{equation}\label{eq:Tk definition}
\begin{aligned}
T_k:=\tr\left(\rho\bigotimes_{i=1}^k\rho_i\right)
\end{aligned}
\end{equation}
based on randomized measurements \cite{van2012Measuring,Elben2018Random,Brydges2019Probing}. 
We denote $T_k$ as the \emph{correlation overlap} (CRO), which is directly relative to the Hilbert-Schmidt distance
\begin{equation}
\begin{aligned}
D_{\mathrm{HS}}\left(\rho,\bigotimes_{i=1}^k\rho_i\right)=\tr(\rho^2)+\prod_{i=1}^k\tr(\rho_i^2)-2T_k.
\end{aligned}
\end{equation}
When $\rho$ is a low-rank state \cite{coles2019tracedis}, such quantity can offer a tight bound of the trace distance between $\rho$ and its marginals, which can be further applied in the quantum independence testing \cite{yu2019quantum}. In addition, we also discuss the application of bipartite CRO in bipartite entanglement detection, and leave it in Appendix \ref{app:T2EW}.

\section{Efficient Estimation Protocols} \label{sec:total correlation}
We now show that the total correlation defined in the previous section can be effectively estimated, irrelevant of the party number $k$. For simplicity of discussion, take the tripartite state $\rho_{ABC}$ as an example. Following the definition in Sec.~\ref{sec:def}, the essential quantity one needs to evaluate is the tripartite CRO
\begin{equation}\label{eq:T3}
\begin{aligned}
T_3=\tr\left[\rho_{ABC}(\rho_A\otimes\rho_B\otimes\rho_C)\right].
\end{aligned}
\end{equation}
The difficulty to measure $T_3$ lies in that it is a nonlinear function of $\rho_{ABC}$ and thus cannot be obtained by measuring the observable on a single-copy state. In fact, given four identical states $\rho_{ABC}^{\otimes 4}$, one can make a joint measurement among these copies \cite{Ekert_2002},
\begin{equation}\label{eq:global obs}
\begin{aligned}
T_3&=\tr\left\{S_A\otimes S_B\otimes S_C\left[\rho_{ABC}\otimes(\rho_A\otimes\rho_B\otimes\rho_C)\right]\right\}\\
&=\tr\left[S_A^{(1,2)}\otimes S_B^{(1,3)}\otimes S_C^{(1,4)}(\rho_{ABC}^{\otimes 4})\right].
\end{aligned}
\end{equation}
Here $S_A^{(1,2)}$ is the {\footnotesize SWAP} operator acting on the Hilbert space of the first two copies of subsystem $A$, $\mathcal{H}_A^1\otimes\mathcal{H}_A^2$, and acts trivially on the last two copies, $S_A^{(1,2)}|\psi\rangle_A^1|\phi\rangle_A^2=|\phi\rangle_A^1|\psi\rangle_A^2$. And similarly for the other {\footnotesize SWAP} operators $S_B^{(1,3)}$ and $ S_C^{(1,4)}$ (see Fig.~\ref{fig:tripartite correlation}(c) for an illustration).

This kind of measurement in general demands the preparation of identical copies of the state $\rho$, and the joint measurements across the distinct copies, which is possible for the one-dimensional system and the few parties case, for example, $k=2$ \cite{Islam2015Measuring,Kaufmanen2016tanglement}. However, it is very challenging for the higher-dimensional system and for the number of parties $k$ being not small. In the following, we develop a measurement protocol based on randomized measurements \cite{van2012Measuring,Elben2018Random,Brydges2019Probing,huang2020predicting}, which only needs the preparation of singlecopies of the state $\rho$. Randomized measurements find applications not only in quantum information, like entanglement negativity extraction \cite{elben2020mixedstate,singlezhou,neven2021symmetryresolved}, entanglement detection \cite{Tran2016Correlations,ketterer2019characterizing,Knips2020,ketterer2020entanglement,ketterer2020certifying}, Fisher information quantification \cite{rath2021Fisher,yu2021Fisher}, and quantum certification \cite{Elben2020Cross,zhang2020experimental,zhang2021experimental}, but also in quantum many-body physics \cite{Vermersch2019Scrambling,Elben2020topological,Cian2020Chern,garcia2021quantum}.

\textbf{Global Measurement Protocol} -- We first propose a means to measure CRO using random unitary gates that act globally on each system. This protocol can then be subsequently modified to use only local unitary gates on each qubit with a modest sacrifice in error scaling. Our global measurement protocol works as follows: Sample and operate random unitary $U=\bigotimes_{i=1}^{k}U_{g_i}$ on each subsystem $g_i$ for the total $k$-partite system, independently, and then conduct computational basis measurement $\ket{s}=\ket{s_{g_1},s_{g_2},\dots,s_{g_k}}$ in a sequential manner. After sufficient repeating of the preparation and measurement, one can get the estimation of the conditional probability 
\begin{equation}\label{eq:globalCprob}
\begin{aligned}
\mathrm{Pr}\left(s_{g_1},s_{g_2},\dots,s_{g_k}\Bigg|\bigotimes_{i=1}^{k}U_{g_i}\right)=\bra{s}U\rho U^{\dag}\ket{s}
\end{aligned}
\end{equation}
and also its marginals $\mathrm{Pr}(s_{g_i}|U_{g_i})$ for the $i$-th subsystem. 
The target quantity CRO $T_k$, can be written as the postprocessing of these conditional probabilities shown in Proposition \ref{prop:Tkglobal}, and we summarize the protocol in Algorithm \ref{algo:correlation measure}.

\begin{algorithm}[H]
\caption{Global Measurement Protocol for $T_k$}\label{algo:correlation measure}
\begin{algorithmic}[1]
\Require
$N_U\times N_M$ sequentially prepared $\rho$
\Ensure
Probability distribution of the measurement outcomes conditioned on the evolution unitary $\mathrm{Pr}\left(s_{g_1},s_{g_2},\dots,s_{g_k}\Big|\bigotimes_{i=1}^{k}U_{g_i}\right)$ in Eq.~\eqref{eq:globalCprob}.

\For{$i= 1~\text{\textbf{to}}~N_U$} 
 \State Randomly pick a unitary matrix $U=\bigotimes_{i=1}^{k}U_{g_i}$, with each $U_{g_i}\in \mc{H}_{g_i}$ sampled uniformly from the unitary 2-design ensemble.
 \State Operate $U$ on $\rho$ to get $U \rho U^\dag$. 
 \For{$j= 1~\text{\textbf{to}}~N_M$} 
  \State  Measure $U\rho U^\dagger$ in the computational basis $\{\ket{s}=\ket{s_{g_1},s_{g_2},\dots,s_{g_k}}\}$.
  \State Record the measurement results.
  \EndFor
 \State Estimate the probability and its marginals in Eq.~\eqref{eq:globalCprob}.
\EndFor
\State Do the data postprocessing given in Proposition \ref{prop:Tkglobal} for $T_k$.
\end{algorithmic}
\end{algorithm}

\begin{proposition}\label{prop:Tkglobal}
For a $k$-partite state $\rho$, the CRO $T_k$ defined in Eq.~\eqref{eq:Tk definition}, can be evaluated by postprocessing the measurement data, i.e., averaging the multiplication of the total and the marginal probabilities under the random unitary evolution as follows:
\begin{equation}\label{eq:core eq correlation}
\begin{aligned}
T_k=\sum_{s,s'}\sub{\mbb{E}}{U}\left[\mathrm{Pr}(s|U)\prod_{i=1}^kX_{\mathop{g}_i}(s_{\mathop{g}_i},s_{\mathop{g}_i}')\ \mathrm{Pr}(s_{\mathop{g}_i}'|U_{g_i})\right],
\end{aligned}
\end{equation}
with the function
\begin{equation} \label{eq:Xg}
\begin{aligned}
X_{\mathop{g}_i}(s_{\mathop{g}_i},s_{\mathop{g}_i}')=-(-d_{\mathop{g}_i})^{\delta_{s_{\mathop{g}_i},s_{\mathop{g}_i}'}},
\end{aligned}
\end{equation}
where $d_{\mathop{g}_i}$ is the dimension of the $i$-th subsystem $g_i$, and $\mbb{E}_U$ denotes averaging over unitary 2-design ensembles on each subsystem $g_i$ independently. 
\end{proposition}
The detailed proof is left in Appendix \ref{proof:Tkglobal}. The intuition is that by multiplying the probabilities in the postprocessing, one can \emph{virtually} get a few copies of $\rho$. Then by averaging on the random unitary, one can further generate permutation operators among virtual copies.  

\begin{figure*}[htbp]
    \centering
    \includegraphics[scale=0.3]{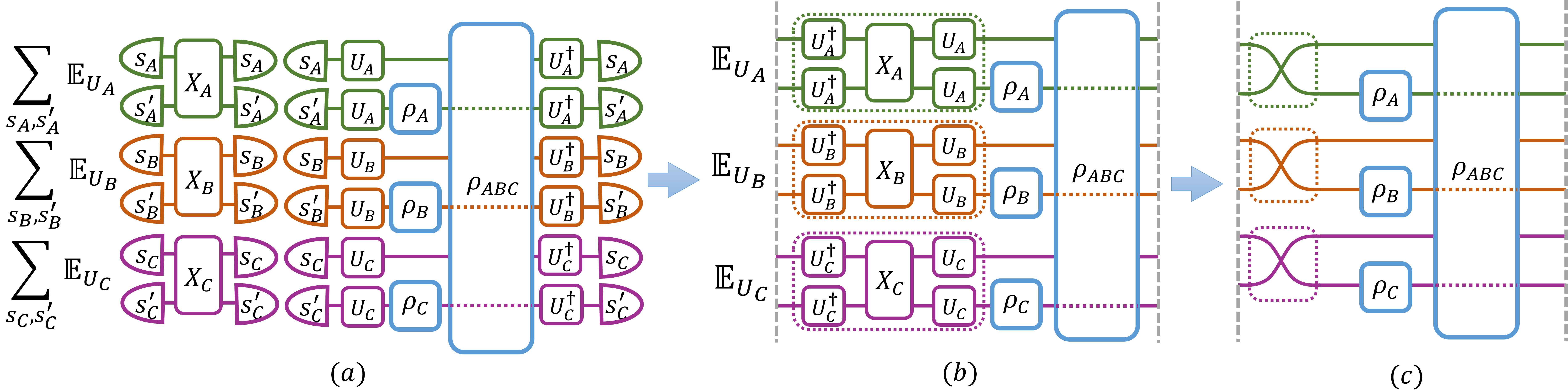}
    \caption{Diagrammatic illustration of the proof of Proposition \ref{prop:Tkglobal}. Here, for simplicity, we take the tripartite state as an example. (a) is the diagram representation of Eq.~\eqref{eq:core eq correlation}. The small half circles with $s,s'$ inside represent the computational basis measurement. The dashed lines indicate that $\rho_{ABC}$ is not connected to $\rho_A\otimes\rho_B\otimes\rho_C$. In (b), we use the cyclic property of the trace formula to put the random unitary evolution on $X_g$, for $g\in \{A,B,C\}$, defined in Eq.~\eqref{eq:Xg}. The vertical gray dashed lines denote the periodic boundary condition, i.e., the trace. The colored dashed boxes are twofold twirling channels acting on $X_g$, which equal to the {\footnotesize SWAP} operators as shown in Eq.~\eqref{eq:SWAP construction}. As a result, we recover $T_3$ in (c) with these {\footnotesize SWAP} operators represented by the ``X''-shape cross in each dashed box, with the formula given in Eq.~\eqref{eq:global obs}.}
    \label{fig:tripartite correlation}
\end{figure*}

Here we sketch the proof outline using Fig.~\ref{fig:tripartite correlation} for the $T_3$ of the tripartite state $\rho_{ABC}$. In Fig.~\ref{fig:tripartite correlation}(a), the conditional probabilities are multiplied, and the box labels the classical function $X_g(s_g,s_{g'})$ with $g\in\{A,B,C\}$ for the subsystem. In Fig.~\ref{fig:tripartite correlation}(b), by using the cyclic property of the trace, we can effectively put the random unitary  on the operator $X_g=\sum_{s_g,s_g'}X_g(s_g,s_{g}')\ket{s_g,s_{g}'}\bra{s_g,s_{g}'}$. In Fig.~\ref{fig:tripartite correlation}(c), we average on the unitary ensemble to generate a {\footnotesize SWAP} operator by the identity \cite{Elben2018Random,Elben2019toolbox},
\begin{equation}\label{eq:SWAP construction}
\begin{aligned}
\Phi^2(X_g):=\sub{\mbb{E}}{U_g\in\mathcal{E}}\left[U_g^{\otimes 2}X_gU_g^{\dagger \otimes 2}\right]=S_g.
\end{aligned}
\end{equation}
We denote this average on the two-copy Hilbert space as the ``twrling'' channel $\Phi^2(\cdot)$. Note that the unitary ensemble $\mathcal{E}$ need not be the Haar measure, and any unitary 2-design ensemble (such as the Clifford group \cite{Divincenzo2002quantum,dankert2009exact}) is sufficient,  which is more practical compared with the previous work by some of us \cite{singlezhou}. If one naively generalizes the protocol there to measure the $k$-partite CRO, a $(k+1)$-design ensemble is needed. Unitary $t$-design with $t\geq4$ is still poorly understood \cite{zhu2016clifford}, and the generation of these ensembles would need deep quantum circuit, which is quite impractical compared to the current protocol.

In real experiments, the sampling time $N_U$ and the measurement time $N_M$ as shown in Algorithm \ref{algo:correlation measure} are both finite; the postprocessing will be more delicate compared to Eq.~\eqref{eq:core eq correlation} which corresponds to the case where $N_U$ and $N_M$ are infinite. We show how to construct an unbiased estimator for the scenario with finite $N_U$ and $N_M$ in Sec.~\ref{stat}.
We remark that our protocol measures the fidelity of the state to its marginal, not with an unrelated state as in Ref.~\cite{Elben2020Cross}.   By adequately utilizing the marginal distributions, our postprocessing shown in Sec.~\ref{stat} needs less state samples and thus is more efficient than directly applying the former one, say Ref.~\cite{Elben2020Cross}. 
Our protocol can be further generalized to local gate version, as discussed in the next section.

\textbf{Local Measurement Protocol} -- We can further simplify the procedure above so that it makes use of only single-qubit random gates. Specifically, the global measurement protocol involves the need to sample random unitary on each subsystem, which may contain several qubits. This is challenging even for the moderate subsystem size. In contrast, the following local measurement protocol only involves random single-qubit Pauli measurement.

Recall that the essence of the global measurement protocol is to construct ``virtual'' {\footnotesize SWAP} operators across different copies in Eq.~\eqref{eq:global obs} by data postprocessing shown in Proposition \ref{prop:Tkglobal}. In fact, {\footnotesize SWAP} operator is factorizable. The big {\footnotesize SWAP} operator $S$ acting on $n$-qubit pairs $\mathcal{H}_2^{\otimes n}\otimes\mathcal{H}_2^{\otimes n}$, 
can be decomposed as $S=\bigotimes_{l=1}^nS_l$,
with $S_l$ the small {\footnotesize SWAP} operator for the $i$-th qubit pair (see Fig.~\ref{fig:local method} for an illustration). This fact enlightens us to substitute the random unitary, say $U_A$ (also $U_B$ and $U_C$) in Fig.~\ref{fig:tripartite correlation}, to the tensor product form $U_A=\bigotimes_{l=1}^{n_A}U_l$, where each single-qubit unitary $U_l$ is from the 2-design ensemble independently, similar for other subsystems $B$ and $C$. Correspondingly, the postprocessing function $X_{\mathop{g}}(s_{\mathop{g}},s_{\mathop{g}}')$ in Eq.~\eqref{eq:Xg} is modified to the multiplication of local functions as shown in Eq.~\eqref{eq:Xglocal}.

\comments{
that the global {\footnotesize SWAP} operator can be constructed by averaging over qubit unitary two design:
\begin{equation}\label{eq:qubit swap}
\begin{aligned}
\bigotimes_{i=1}^NS_i&=\bigotimes_{i=1}^N\sub{\mbb{E}}{U_i}\left[U_i^{\otimes 2}\left(\sum_{s_i,s_i'}2(-2)^{-D[s_i,s_i']}|s_i,s_i'\rangle\langle s_i,s_i'|\right)U_i^{\dagger\otimes 2} \right]\\
&=\sub{\mbb{E}}{U}\left[U^{\otimes 2}\left(\sum_{\vec{s},\vec{s}'}2^N(-2)^{-D[\vec{s},\vec{s}']}|\vec{s},\vec{s}'\rangle\langle\vec{s},\vec{s}'|\right)U^{\dagger\otimes 2}\right] .
\end{aligned}
\end{equation}
}

For the general $k$-partite state $\rho$, suppose it contains $n=\sum_{i=1}^k n_{g_i}$ qubits with $i$-th party having $n_{g_i}$ qubits. We denote the computational basis as the $n$-bit binary vector $|\vec{s}\rangle=|s_1,s_2,\dots,s_n\rangle,s_l=0/1$, and the state restricted on the $i$-th party as $\ket{\vec{s}_{\mathop{g}_i}}$.
By modifying the global protocol in Algorithm \ref{algo:correlation measure}, the local measurement protocol of CRO $T_k$ is shown in Algorithm \ref{algo:local}, and now we aim to obtain the following conditional probability
\begin{equation}\label{eq:localCprob}
\begin{aligned}
\mathrm{Pr}\left(s_1,s_2,\dots,s_n\Bigg|\bigotimes_{l=1}^{n}U_l\right)
\end{aligned}
\end{equation}
and its marginals $\mathrm{Pr}(\vec{s}_{\mathop{g}_i}|U_{g_i})$ with $U_{g_i}=\bigotimes_{l\in g_i} U_l$. The data postprocessing is summarized in Proposition \ref{prop:Tklocal}.
\begin{proposition}\label{prop:Tklocal}
Given a $k$-partite state $\rho$, the CRO $T_k$ defined in Eq.~\eqref{eq:Tk definition}, can be evaluated by postprocessing the measurement data, i.e., averaging the multiplication of the total and the marginal probabilities under the single-qubit random unitary evolution as follows.
\begin{equation}\label{eq:Tklocal}
\begin{aligned}
T_k=\sum_{\vec{s},\vec{s}'}\sub{\mbb{E}}{U}\left[\mathrm{Pr}(\vec{s}|U)\prod_i^k\tilde{X}_{\mathop{g}_i}(\vec{s}_{\mathop{g}_i},\vec{s}_{\mathop{g}_i}') \mathrm{Pr}(\vec{s}_{\mathop{g}_i}'|U_{g_i})\right]
\end{aligned}
\end{equation}
with the function
\begin{equation} \label{eq:Xglocal}
\begin{aligned}
\tilde{X}_{\mathop{g}_i}(\vec{s}_{\mathop{g}_i},\vec{s}_{\mathop{g}_i}'):= \prod_{l\in g_i} X_{l}(s_l,s'_l)
\end{aligned}
\end{equation}
where $X_{l}$ is defined in Eq.~\eqref{eq:Xg} with $d_l=2$, and $\mbb{E}_U$ denotes averaging $U=\bigotimes_{l=1}^{n}U_l$ over unitary 2-design ensembles on each qubit independently. 
\end{proposition}
Proposition \ref{prop:Tklocal} can be proved following the proof of Proposition \ref{prop:Tkglobal}, and the proof is left to Appendix \ref{proof:Tklocal}.

\comments{
\begin{proposition}
\begin{equation}\label{eq:qubit measurement}
\begin{aligned}
T_3=\sum_{\vec{s},\vec{s}'}2^N(-2)^{-D[\vec{s},\vec{s}']}\sub{\mbb{E}}{U}\left[P_{ABC}(\vec{s}_A,\vec{s}_B,\vec{s}_C|U_A,U_B,U_C)P_A(\vec{s}_A'|U_A)P_B(\vec{s}_B'|U_B)P_C(\vec{s}_C'|U_C)\right].
\end{aligned}
\end{equation}
Where $\vec{s}=(\vec{s}_A,\vec{s}_B,\vec{s}_C)$, $\vec{s}'=(\vec{s}_A',\vec{s}_B',\vec{s}_C')$, $|\vec{s}\rangle$ is the computational basis of this N-qubit Hilbert space. $\sub{\mbb{E}}{U}$ denotes averaging over N qubit unitary 2-designs.
\end{proposition}
}

\begin{algorithm}[H]
\caption{Local Measurement Protocol for $T_k$}\label{algo:local}
\begin{algorithmic}[1]
\Require
$N_U\times N_M$ sequentially prepared $\rho$
\Ensure
Probability distribution of the measurement outcomes conditioned on the evolution unitary $\mathrm{Pr}\left(s_1,s_2,\dots,s_n\Bigg|\bigotimes_{l=1}^{n}U_l\right)$ in Eq.~\eqref{eq:localCprob}.

\For{$i= 1~\text{\textbf{to}}~N_U$} 
 \State Randomly pick a unitary matrix $U=\bigotimes_{l=1}^{n}U_l$, with each $U_l\in \mc{H}_{l}$ on the $l$-th qubit sampled uniformly from the unitary 2-design ensemble.
 \State Operate $U$ on $\rho$ to get $U \rho U^\dag$. 
 \For{$j= 1~\text{\textbf{to}}~N_M$} 
  \State  Measure $U\rho U^\dagger$ in the computational basis $\{\ket{s}=|s_1,s_2,\dots,s_n\rangle$.
  \State Record the measurement results.
  \EndFor
 \State Estimate the probability and its marginals in Eq.~\eqref{eq:localCprob}.
\EndFor
\State Do the data postprocessing given in Proposition \ref{prop:Tklocal} for $T_k$.
\end{algorithmic}
\end{algorithm}

Besides the practicality of the local measurement protocol, another advantage is that the measurement procedure and the postprocessing procedure are decoupled. In particular, one can choose to study the correlation for any partition of the system or the correlation information restricted on some subsystems in parallel, by only changing the postprocessing function.

\begin{figure}[htbp]
    \centering
    \includegraphics[scale=0.23]{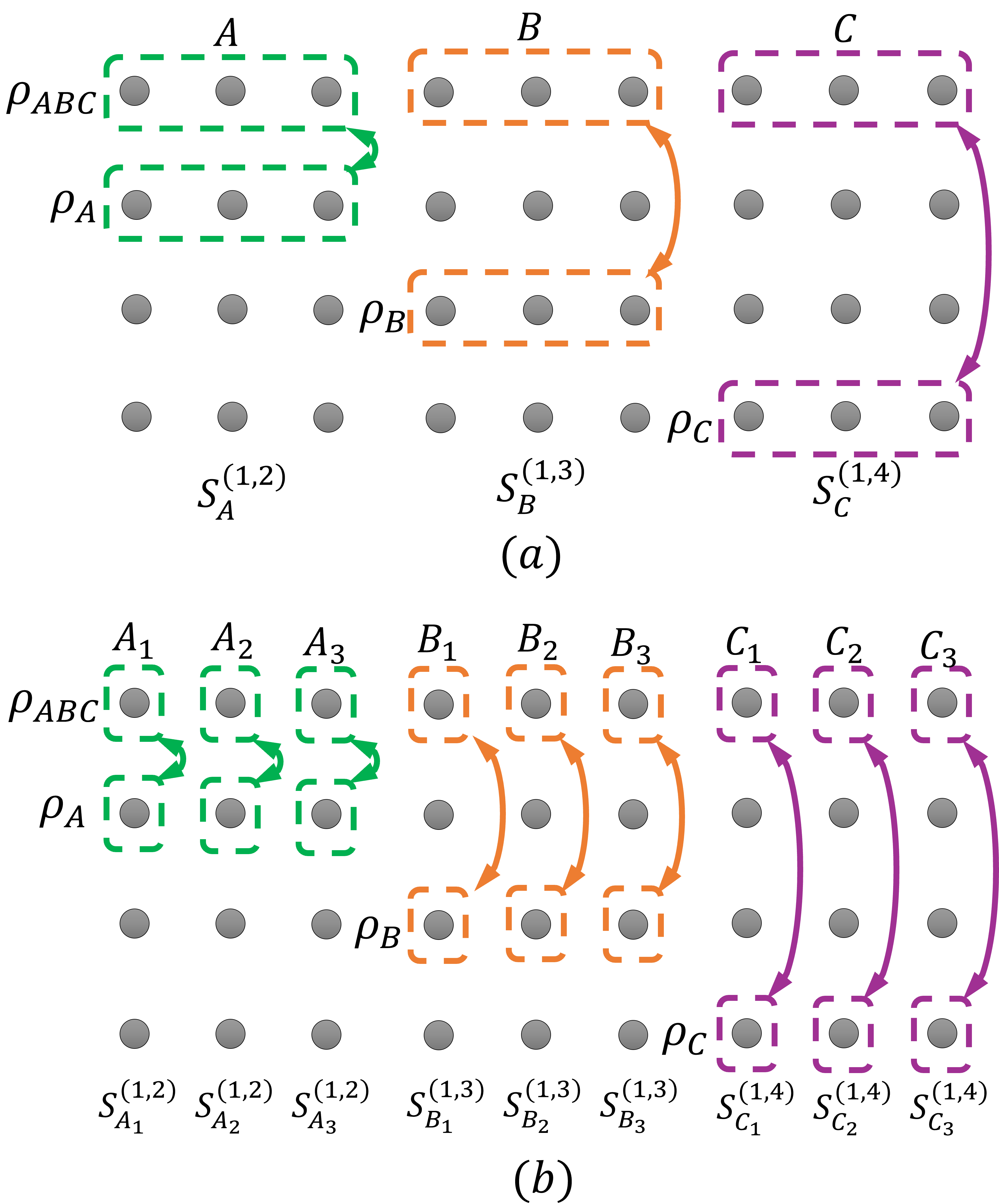}
    \caption{Here we provide a graphical example explaining how the global and local methods work. $\rho_{ABC}$ is an nine-qubit three-partite quantum state with each party containing three qubits. The curved lines with arrows in both ends denote {\footnotesize SWAP} operators. In global method (a), three-qubit {\footnotesize SWAP} operators are constructed using three-qubit random unitary matrix $U_A$, $U_B$, and $U_C$.   While in local method (b), the three-qubit {\footnotesize SWAP} operators are decomposed into a tensor product of single-qubit {\footnotesize SWAP} operators, which can be constructed using single-qubit random unitary gates. Each colored dashed box represents a unitary 2-design.}
    \label{fig:local method}
\end{figure}

\section{Statistical analysis}\label{stat}
In practical situations, the sampling times of the random unitary matrices $N_U$, and the number of the projective measurement $N_M$ under a given unitary are both finite. The multiplication of them, $N_UN_M$, quantifies how many copies of $\rho$ in total one needs to prepare in sequence. 
In this section, for clarity, we focus on the tripartite CRO in Eq.~\eqref{eq:T3} and construct an unbiased estimator for it. We then analyze the variance of the estimator in this finite sampling scenario. The scaling of the variance respective to $N_U$, $N_M$ and $D$ characterizes the sample complexity of our protocol. Similar analysis works for the cases of general $k$-partite CRO. We use $\ket{s}=\ket{s_A,s_B,s_C}$ to denote the computational basis of $\mathcal{H}=\mc{H}_A\otimes \mc{H}_B\otimes \mc{H}_C$, where the total dimension is $D=d_A d_B d_C$. In the following, we first illustrate the estimation protocol with global unitary evolution, and then proceed to show how this can be converted to the one with the single-qubit measurements.

To construct an unbiased estimator, we first note that the postprocessing expression in Eq.~\eqref{eq:core eq correlation} can be equivalently written as $4$-time multiplication of the probability distribution, 
\begin{equation} \label{eq:core eq correlation2}
\begin{aligned}
T_3 = \sum_{ \bm{s}_A,\bm{s}_B,\bm{s}_C } &X^{(1,2)}_A(s^1_A,s^2_A) X^{(1,3)}_B(s^1_B,s^3_B)  X^{(1,4)}_C(s_C^1,s_C^4)\\ 
&\times\sub{\mbb{E}}{U}\left[\prod_{i=1}^4 \pr(s_A^i, s_B^i, s_C^i|U_A, U_B, U_C)\right],
\end{aligned}
\end{equation}
where we denote $\bm{s}_A=(s_A^1,s_A^2,s_A^3,s_A^4)$ as a 4-dit string with $s_A^i\in \{0,1,\dots,d_A-1\}$, similar for $\bm{s}_B$ and $\bm{s}_C$. $X^{(1,2)}_A(s_A^1,s_A^2)$ is the function in Eq.~\eqref{eq:Xg} restricted on the first two indices, 
similar for $X_B^{(1,3)}(s_B^1,s_B^3)$ and $X_C^{(1,4)}(s_C^1,s_C^4)$. 

In Algorithm~\ref{algo:correlation measure}, one samples $N_U$ times of $U=U_A\otimes U_B\otimes U_C$ to perform experiments. For the $t$-th round of unitary sampling, one repeats the preparation and measurement for $N_M$ times. For the $i$-th time of measurement, we define a matrix-valued random variable 
\begin{equation}
    \hat{r}_U(i)=\ket{\hat{s}_U(i)}\bra{\hat{s}_U(i)}, 
\end{equation}
where $\hat{s}_U(i)$ is a classical random variable with the conditional probability
\begin{equation}
    \pr(s|U)=\bra{s}U\rho_{ABC} U^{\dag}\ket{s},
\end{equation}
to record the measurement result $s$. For each random unitary choice, one finally gets $N_M$ independent samples $\{\hat{r}_U(i)\}_{i=1}^{N_M}$. We then construct an estimator for $T_3$ as follows
\begin{widetext}
\begin{equation} \label{eq:Mt}
\begin{aligned}
    \hat{M}(t) &= \binom{N_M}{4}^{-1} \sum_{1\leq i<j<k<l \leq N_M} \tr\left\{ Q_3 \left[\hat{r}_U(i) \otimes \hat{r}_U(j) \otimes \hat{r}_U(k) \otimes \hat{r}_U(l)\right] \right\} \\
    &= \binom{N_M}{4}^{-1} \sum_{1\leq i<j<k<l\leq N_M}  X^{(1,2)}_A\left(\hat{s}_U(i),  \hat{s}_U(j)\right)   X^{(1,3)}_B\left( \hat{s}_U(i), \hat{s}_U(k) \right)  X^{(1,4)}_C\left( \hat{s}_U(i), \hat{s}_U(l) \right), \\
\end{aligned}
\end{equation}
\end{widetext}
with 
\begin{equation}\label{eq:processQ}
    Q_3 := \left( X_A^{(1,2)}\otimes I_A^{(3,4)} \right) \otimes \left( X_B^{(1,3)}\otimes I_B^{(2,4)} \right) \otimes \left( X_C^{(1,4)}\otimes I_C^{(2,3)} \right)
\end{equation}
being an observable on $\mc{H}^{\otimes 4}$. $\hat{M}(t)$ is an unbiased estimator in the sense that $\mathbb{E}_{U,\mb{s}} \left[\hat{M}(t) \right]=T_3$, with the expectation value taken for all random $U$ and measurement outputs. Since the estimators $\{\hat{M}(t)\}_{t=1}^{N_U}$ are 
independent and identically distributed, the final estimator is defined as $\hat{M}=\frac1{N_U} \sum_{t=1}^{N_U} \hat{M}(t)$, which is naturally unbiased. 

In the local measurement protocol, the unitaries on subsystems $A$, $B$ and $C$ are substituted to products of the random unitaries on qubits. To construct the unbiased estimator for the local protocol, accordingly the postprocessing matrix $Q$ in Eq.~\eqref{eq:processQ} should be adjusted to 
\begin{equation}
    Q_{3,\mathrm{loc}} := \left( \tilde{X}_A^{(1,2)}\otimes I_A^{(3,4)} \right) \otimes \left( \tilde{X}_B^{(1,3)}\otimes I_B^{(2,4)} \right) \otimes \left( \tilde{X}_C^{(1,4)}\otimes I_C^{(2,3)} \right)
\end{equation}
with $\tilde{X}_A^{(1,2)}=\bigotimes_{i=1}^{n_A} X_{A_i}^{(1,2)}$ the product of the qubitwise $X$ operator. Similar as in Eq.~\eqref{eq:Mt}, one can construct the final unbiased estimator $\hat{M}_L=\frac1{N_U} \sum_{t=1}^{N_U} \hat{M}_L(t)$.

To construct the unbiased estimator for $T_k$, one just needs to extend the definition of $Q_3$ and $Q_{3,\mathrm{loc}}$ to the $k$-partite scenario
\begin{equation}
\begin{aligned}
Q_k &:= \bigotimes_{i=2}^{k+1}\left(X_{g_i}^{(1,i)}\otimes I_{g_i}^{\overline{(1,i)}}\right),\\
Q_{k,\mathrm{loc}} &:= \bigotimes_{i=2}^{k+1}\left(\tilde{X}_{g_i}^{(1,i)}\otimes I_{g_i}^{\overline{(1,i)}}\right),
\end{aligned}
\end{equation}
where $\overline{(1,i)}$ is the complementary set of $(1,i)$ of $\{1,2,\dots,k+1\}$.
We further give the following result on the variance of these constructed estimators for $T_k$.
\begin{proposition}\label{prop:var}
In the regime $D\gg N_M\gg k$, the variance of the unbiased estimators $\hat{M}$ and $\hat{M}_L$ for the $k$-partite CRO show the following scaling:
\begin{equation} \label{eq:varscale}
\begin{aligned}
&\mathrm{Var}\left(\hat{M}\right) =\Theta(\frac{D}{N_UN_M^{k+1}}),\\
&\mathrm{Var}\left(\hat{M}_L\right) =O(\frac{D^{\log_23}}{N_UN_M^{k+1}}),
\end{aligned}
\end{equation}
where $\hat{M}$ and $\hat{M}_L$ are constructed with the measurement data from the protocols in Algorithm \ref{algo:correlation measure} and Algorithm \ref{algo:local}, respectively.
\end{proposition}

For the global random unitary case, we rigorously prove that the variance scales linearly with $D$, $N_U^{-1}$, and $N_M^{-(k+1)}$; while for the local random unitary case, we provide an upper bound on the scaling. No matter in the global or the local case, such error scaling is much better than full tomography \cite{haah2017sample}. Inaddition, the variance decreases when increasing the party number $k$, which is equivalent to the number of virtual copies of state. 
In Appendices \ref{app:variance1} and \ref{app:variance2}, we provide a detailed analysis of the statistical variance.

To support our theoretical analysis, we conduct numerical experiments for the local protocol, i.e., the random unitary matrix applied is the tensor product of the random qubit ones. The numerical results are shown in Fig.~\ref{fig:numerical data}. In Fig.~\ref{fig:numerical data}(a)-\ref{fig:numerical data}(c), we choose the tripartite Greenberger-Horne-Zeilinger (GHZ) state, with an equal qubit number in each party, as the target state. The exact value $T_3=0.125$ is independent of the qubit number, so that the variance itself is suitable to quantify the quality of the estimation result. We first show how the variance changes with $N_U$ when measuring a three-qubit GHZ state for different $N_M$ in Fig.~\ref{fig:numerical data}(a). These three lines with slopes about $-1$ are coincident with the conclusion in Proposition\ref{prop:var} that $\mathrm{Var}\left(\hat{M}_L\right)\propto N_U^{-1}$. The variance decreases with the increase of $N_M$, the measurement times per unitary evolution. Then, by adjusting the qubit number of the GHZ state and with a fixed $N_M=10$, we also find that the variance increases for larger system dimension in Fig.~\ref{fig:numerical data}(b).  

To study how the variance scales with dimension $D$ when $D\gg N_M$, we change the qubit number of the target GHZ state from 6 to 21 and set $N_U=100$ and $N_M=10$ in Fig.~\ref{fig:numerical data}(c). The slope $\alpha=1.2587$ from the linear regression of $\log_2\left[\mathrm{Var}\left(\hat{M}_L\right)\right]$ and the qubit number $n$. It indicates $\mathrm{Var}\left(\hat{M}_L\right)\propto D^\alpha$ with $\alpha \approx 1.26<\log_23$, which is consistent with our theoretical result in Proposition\ref{prop:var}.  In Fig.~\ref{fig:numerical data}(d), we take a six-qubit noisy $W$ state as an example, show the measurement results for different $N_M$ with $N_U=100$, 
and find that our protocol can provide such high-quality measurement results as $N_M \geq 20$.

\begin{figure}[htbp]
\centering

\includegraphics[width=8.5cm]{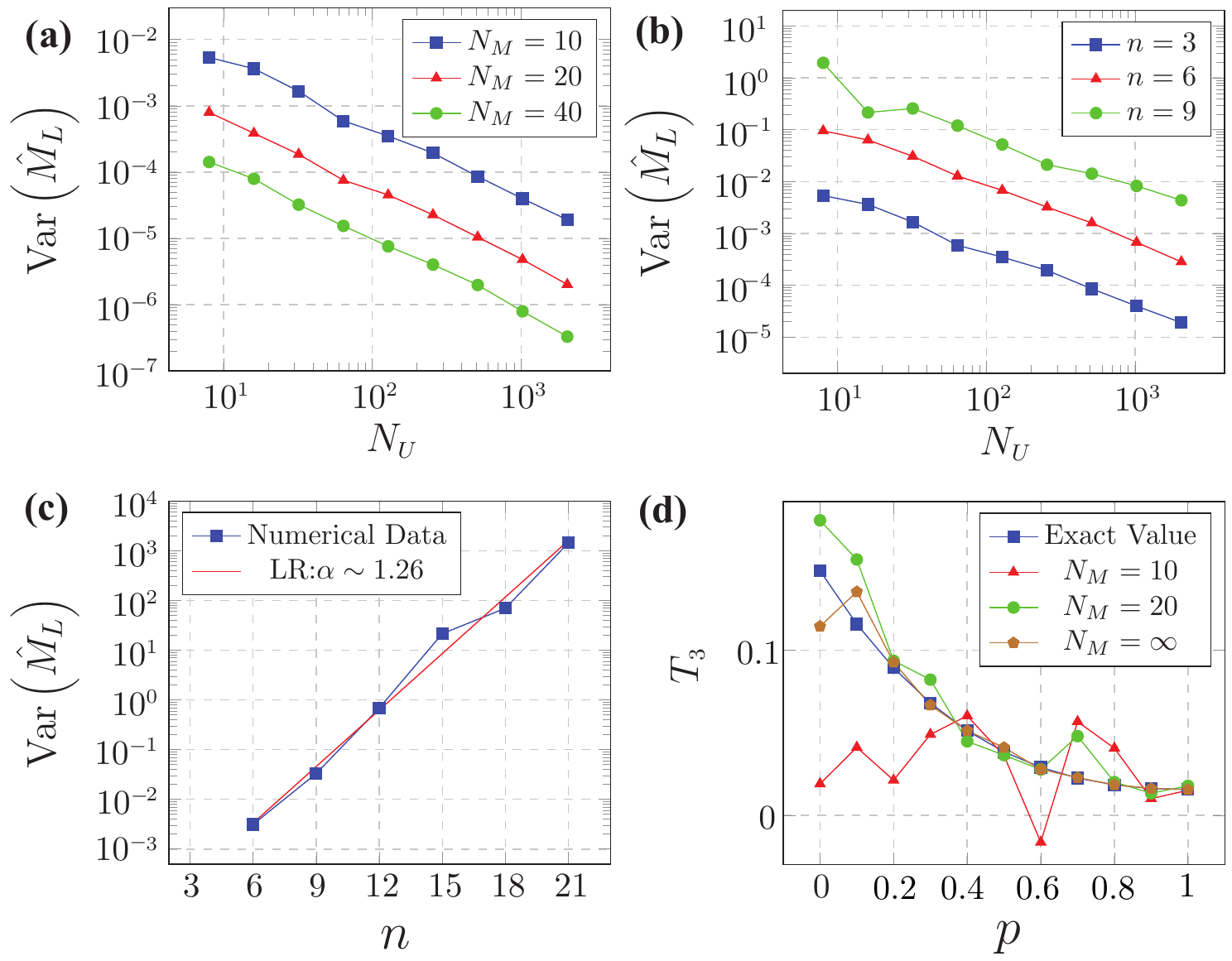}

\caption{Numerical results for the estimation of the tripartite CRO $T_3$ in Eq.~\eqref{eq:T3} with the local measurement protocol. 
(a) The variance scaling with $N_U$ for different $N_M$ when measuring the three-qubit GHZ state. (b) The variance scaling with $N_U$ when measuring the GHZ state with different number of qubits, and $N_M=10$. (c) The variance dependence on the number of qubits of the measured GHZ state with $N_U=100$ and $N_M=10$. We linearly regress the data and obtain the slope $\alpha=1.2587$. (d) The estimation for the noisy state. We measure $T_3$ of a six-qubit noisy $W$ state $\rho_{ABC}=(1-p)|W\rangle\langle W|+\frac{p}{2^6}I$ with $N_U=100$ and different $N_M$.}
\label{fig:numerical data}
\end{figure}

\section{Application to Measuring Fidelity to Maximally Entangled States} \label{sec:fide}
Our protocol can also be modified to measure the fidelity between the candidate bipartite state, and that of a maximally entangled state, 
\begin{equation}
\begin{aligned}
|\Psi^+\rangle=\frac{1}{\sqrt{d}}\sum_{i=1}^d|ii\rangle,
\end{aligned}
\end{equation}
on $\mathcal{H}_A\otimes\mathcal{H}_B$. This fidelity is important in entanglement detection \cite{GUHNE2009detection} and many quantum communication tasks \cite{horodecki2009quantum}. Prior methods need an ideally prepared state \cite{Elben2020Cross} or fixed basis measurements. 
Utilizing random measurements, the maximally entangled state can be virtually produced by postprocessing, and randomness may make it more robust against the noise in the measurement basis.

Without loss of generality, we consider $d=2^n$ with $n$ being the qubit number of each party. 
Recall that the outer product form of the {\footnotesize SWAP} operator is $S=\sum_{i,j=1}^d|i\rangle\langle j|\otimes |j\rangle\langle i|$, and the maximally entangled state is proportional to the partial transpose of $S$ as follows.
\begin{equation}\label{eq:MES and SWAP}
\begin{aligned}
|\Psi^+\rangle\langle\Psi^+|=&\frac{1}{d}\sum_{i,j=1}^d|i\rangle\langle j|\otimes |i\rangle\langle j|\\ 
=&\frac{1}{d}\left(\sum_{i,j=1}^d|i\rangle\langle j|\otimes |j\rangle\langle i|\right)^{T_B}=\frac{1}{d}S^{T_B},
\end{aligned}
\end{equation}
and the corresponding diagrams are shown in Figs.~\ref{fig:random fidelity}(a) and \ref{fig:random fidelity}(b).

Suppose the state we actually produce is $\rho\in\mathcal{H}_A\otimes\mathcal{H}_B$. According to Eq.~\eqref{eq:MES and SWAP}, the fidelity between $\rho$ and $|\Psi^+\rangle$ can be represented using $S$ as
\begin{equation}
\begin{aligned}
\tr\left(\rho|\Psi^+\rangle\langle\Psi^+|\right)=\frac{1}{d}\tr\left(\rho S^{T_B}\right).
\end{aligned}
\end{equation}
Recall that Eq.~\eqref{eq:SWAP construction} shows that the {\footnotesize SWAP} operator can be effectively generated by randomized measurements. Hence the fidelity can be further rewritten as
\begin{equation}\label{eq:random fidelity}
\begin{aligned}
\tr\left(\rho|\Psi^+\rangle\langle\Psi^+|\right)=&\frac{1}{d}\tr\left(\rho[\Phi^2(X)]^{T_B}\right)\\
=&\frac{1}{d}\tr\left\{\rho\sub{\mbb{E}}{U}\left[(U\otimes U)^\dag X(U\otimes U)\right]^{T_B}\right\}\\
=&\frac{1}{d}\tr\left\{\sub{\mbb{E}}{U}\left[(U\otimes U^*)\rho(U^\dagger\otimes U^T)\right]X\right\}.
\end{aligned}
\end{equation}
In the final equality, we use the fact that $X$ is a diagonal matrix such that $X^{T_B}=X$, and also the cyclic property of trace to put the unitary evolution on the state (see Fig.~\ref{fig:random fidelity} for an illustration).

Similar as for the local unitary protocol shown in Sec.~\ref{sec:total correlation}, one can apply the local random unitary 
$U=\bigotimes_{i=1}^nU_i$ here. As a result, the postprocessing matrix is substituted by $\tilde{X}(\vec{s}_A,\vec{s}_B)=\prod_{i} X_{i}(s^A_i,s^B_i)=2^n(-2)^{-D[\vec{s}_A,\vec{s}_B]}$ as in Eq.~\eqref{eq:Xglocal}. Note that here the postprocessing function is on the measurement result of the two parties $A$ and $B$, not on the different copies as before. 
We summarize the postprocessing under local random unitary as follows.
\begin{figure}[htbp]
    \centering
    \includegraphics[scale=0.3]{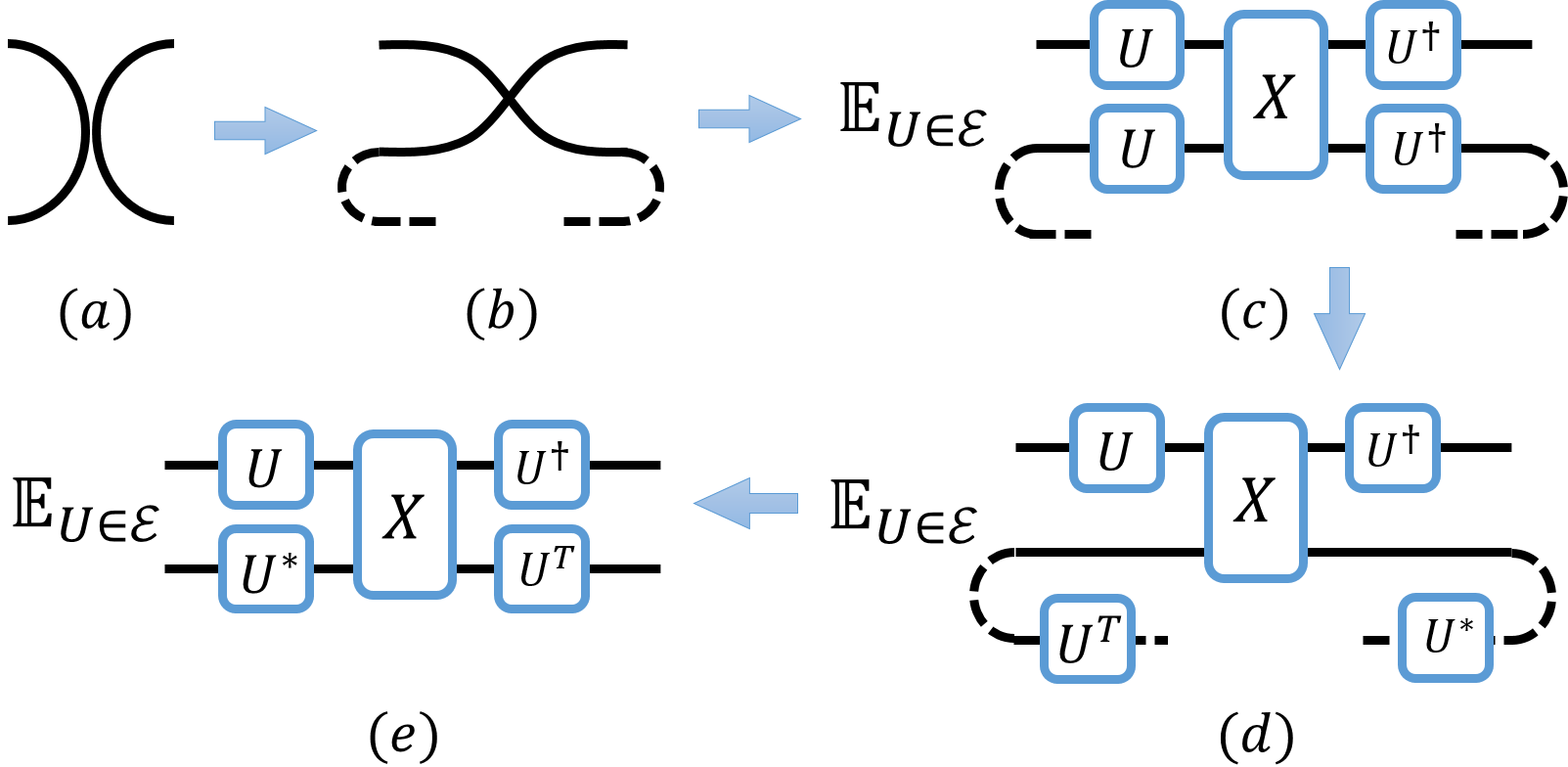}
    \caption{Diagrammatic illustration of the derivation in Eq.~\eqref{eq:random fidelity}. $\mbb{E}_{U\in\mathcal{E}}$ denotes averaging over (qubit) unitary 2-design and the dotted lines represent the transposition on the second subsystem. (a) $d|\Psi^+\rangle\langle\Psi^+|\to$ (b) $S^{T_B}\to$ (c) $[\Phi^2(X)]^{T_B}=\sub{\mbb{E}}{U}\left[(U\otimes U)X(U\otimes U)^\dagger\right]^{T_B}\to$ (d) $\sub{\mbb{E}}{U}\left[(U\otimes U^*)X^{T_B}(U^\dagger\otimes U^T)\right]\to$ (e) $\sub{\mbb{E}}{U}\left[(U\otimes U^*)X(U^\dagger\otimes U^T)\right]$.}
    \label{fig:random fidelity}
\end{figure}

\begin{proposition}  \label{prop:fidelity}

For a bipartite state $\rho_{AB}$ with subsystems A and B both containing $n$ qubits, the fidelity of a state $\rho_{AB}$ with the maximally entangled state $\ket{\Psi^+}$ can be expressed by the following local randomized measurement result:
\begin{equation}
\begin{aligned}
\tr\left(\rho_{AB}|\Psi^+\rangle\langle\Psi^+|\right)=\sum_{\vec{s}_A,\vec{s}_B}(-2)^{-D[\vec{s}_A,\vec{s}_B]}\sub{\mbb{E}}{U}\pr(\vec{s}_A,\vec{s}_B|U),\\
\pr(\vec{s}_A,\vec{s}_B|U)=\tr\left[(U\otimes U^*)\rho(U^\dagger\otimes U^T)|\vec{s}_A,\vec{s}_B\rangle\langle\vec{s}_A,\vec{s}_B|\right],
\end{aligned}
\end{equation}
where $P(\vec{s}_A,\vec{s}_B|U)$ is the probability when measuring  $(U\otimes U^*)\rho(U^\dagger\otimes U^T)$ in the computational basis $\{|\vec{s}_A,\vec{s}_B\rangle\}$;
\comments{
\begin{equation}
\begin{aligned}
\pr(\vec{s}_A,\vec{s}_B|U)=\tr\left[(U\otimes U^*)\rho(U^\dagger\otimes U^T)|\vec{s}_A,\vec{s}_B\rangle\langle\vec{s}_A,\vec{s}_B|\right].
\end{aligned}
\end{equation}
}
the random unitary $U = \bigotimes_{i=1}^{n} U_i$ is a tensor product of unitaries on each qubit, where each $U_i$ is sampled from a unitary $2$-design.
\end{proposition}

Based on Proposition \ref{prop:fidelity}, we summarize the fidelity measurement protocol in Algorithm~\ref{algo:fidelity measurement}. 
\begin{algorithm}[H]
\caption{Fidelity with Local Measurement}
\label{algo:fidelity measurement}
\begin{algorithmic}[1]
\Require
$N_U\times N_M$ sequentially prepared
$\rho_{AB}$
with subsystems $A$ and $B$ both containing $n$ qubits.
\Ensure
Probabilities measured in computational measurement basis under random unitary evolution $\pr(\vec{s}_A,\vec{s}_B|U)$.
\For{$i= 1~\text{\textbf{to}}~N_U$} 
 \State Randomly pick a unitary matrix $U = \bigotimes_{i=1}^{n} U_i$, where $U_i$ on each qubit forms a unitary $2$-design. 
 Operate $U\otimes U^*$ on $\rho$ to get $(U\otimes U^*)\rho (U^\dagger\otimes U^T)$. 
 \For{$j= 1~\text{\textbf{to}}~N_M$} 
  \State  Measure $(U\otimes U^*)\rho (U^\dagger\otimes U^T)$ in the computational basis $\{|\vec{s}_A,\vec{s}_B\rangle\}$.
  \State Record the measurement results.
  \EndFor
 \State Estimate the probabilities $\pr(\vec{s}_A,\vec{s}_B|U)$.
  
\EndFor
\State Do the data postprocessing according to Proposition \ref{prop:fidelity}.
\end{algorithmic}
\end{algorithm}

In Appendix \ref{sec:concurrence}, we also extend the randomized measurement method to estimate the concurrence \cite{wootters2001entanglement,Beacom2004Spectroscopy} of an $n$-qubit quantum state. This shows the broad application scenarios of the randomized measurements.

\section{Discussion}\label{sec:conclusion}
In this work, we introduce an operationally meaningful quantifier of the total correlation within a multipartite quantum system motivated by experimental accessibility. Based on this definition, we design a protocol to estimate the total correlation of a candidate state using only classical postprocessing of data collected from randomized single-qubit measurements, and show that the number of measurements required is significantly lower than that of the state tomography. Taken together, the results provide an accessible tool for characterizing multipartite correlations in NISQ devices.


There are a number of interesting future directions. One direction involves observing that shadow estimation offers an alternative way to postprocess the measurement data under random unitary evolution \cite{huang2020predicting}. Recently, there are enhancements of the error scaling of the shadow protocol by using prior knowledge of the observable \cite{hadfield2020measurements,huang2021efficient,wu2021overlapped}, or the intrinsic tensor-product structure of the underlying state for the nonlinear function estimation \cite{garcia2021quantum}. It would be interesting to ascertain if these methodologies could provide further enhancement to estimating 
the total correlation measurement here, in situations where one has additional knowledge of the NISQ device \cite{rath2021importance}.

The total correlation has many proposed applications. A recent framework for characterizing fine-grained structure or genuine multipartite correlation, for example, involves measuring how correlation changes depending on how one partitions the whole system \cite{Girolami_2017,Bennett_2011}. Meanwhile, such correlation measure could be used as the cost function in the near-term variational algorithms to decouple the quantum system \cite{yuan2019theory,cerezo2020variational,Khatri2019quantumassisted,zhang2020mutual}. Both scenarios would require many costly repeated calls to estimate the correlation. Thus, a natural direction then is to investigate if our techniques provide the reduction to this cost. Meanwhile, many occasions invoke interest in specific types of correlations, such as those that are classical, or purely quantum mechanical, which is also interesting to further investigate with randomized measurements.


\section{acknowledgments}
We thank Arthur Jaffe and Xiongfeng Ma for the useful discussion. This research is supported by the Quantum
Engineering Program QEP-SF3, National Research Foundation of Singapore under its NRF-ANR joint program
(NRF2017-NRF-ANR004 VanQuTe), the Singapore Ministry of Education Tier 1 grant RG162/19, FQXi-RFP-IPW-1903 from the foundational Questions Institute and Fetzer Franklin Fund, a donor advised fund of Silicon Valley Community Foundation, the National Natural Science Foundation of China Grants No.~11875173 and No.~1217040781, and the National Key Research and Development Program of China Grants No.~2019QY0702 and No.~2017YFA0303903. Any opinions, findings and conclusions or recommendations expressed in this material are those of the author(s) and do not reflect the views of the National Research Foundation, Singapore.

\bibliographystyle{apsrev4-2}

%

\appendix

\onecolumngrid
\newpage

\section{Preliminaries}\label{sec:pre}

\subsection{Integral of random unitary matrix}\label{subsec:random unitary}

A random unitary matrix is a random variable in the space of a unitary matrix~\cite{zyczkowski1994random}; Haar measure means that the probability distribution is uniform. 
Based on the definition of Haar-measured random unitary matrix, here we introduce a $t$-fold twirling channel,
\begin{equation}
\begin{aligned}
\Phi^t(O)=\int_{\mathrm{Haar}}dU U^{\otimes t}O U^{\dagger \otimes t},
\end{aligned}
\end{equation}
where $U\in\mathcal{H}_d$ and $O$ is a linear operator acting on $\mathcal{H}_d^{\otimes t}$. According to Schur-Weyl duality~\cite{goodman2000representations,Kliesch2021Certification}, such twirling channel is equivalent to the operation that projects $O$ into the symmetric subspace. So we have 
\begin{equation} \label{eq:twirling channel}
\begin{aligned}
\Phi^t(O)=\int_{\mathrm{Haar}}dU U^{\otimes t}O U^{\dagger \otimes t}=\sum_{\pi,\sigma\in\mathcal{S}_t}C_{\pi,\sigma}\tr(W_\pi O)W_\sigma.
\end{aligned}
\end{equation}
Where $\mathcal{S}_t$ is the $t$-th order permutation group, $C_{\pi,\sigma}$ is the element of the Weingarten matrix \cite{Gu2013Moments}, and $W_\pi$ is the permutation operator corresponding to $\pi$. By adjusting $O$, one can generate any permutation operators; this is the core idea of the estimation protocol proposed in Sec.~\ref{sec:total correlation}.
\comments{Without HRU, permutation operators are usually produced by the interaction among some identical copies of the real physical system \cite{Ekert_2002,Islam2015Measuring}. For a practical physical platform, such method will be extremely hard with the increasing of the qubit number. However, Eq.~\eqref{eq:twirling channel} enlightens us that permutation operator could be produced alternatively.}

Generally speaking, it is impractical to randomly pick an element in unitary space. Fortunately, it has been proved that a $t$-fold twirling channel can be realized by averaging over a unitary ensemble $\mathcal{E}$ within which the elements are all fixed unitary matrices, which we call $\mathcal{E}$ unitary $t$-design \cite{lindner2017design}.
\begin{equation}
\begin{aligned}
\Phi^t_{\mathcal{E}}(O)=\frac{1}{|\mathcal{E}|}\sum_{U\in\mathcal{E}}U^{\otimes t}O U^{\dagger\otimes t}=\sub{\mbb{E}}{U\in\mathcal{E}}\left(U^{\otimes t}O U^{\dagger\otimes t}\right)=\Phi^t(O),
\end{aligned}
\end{equation}
where $|\mathcal{E}|$ denotes the size of $\mathcal{E}$. This fact further reduces the difficulty of implementing twirling channel in real experiments. It is worth mentioning that the Clifford group is a unitary $3$-design and a unitary $t$-design is also an unitary $m$-design with $m<t$.

\subsection{Tensor network basics}\label{subsec:tensor}

Tensor network is one of the graphical methods helping us to deal with tensor calculation \cite{wood2011tensor,shapourian2021entanglement}. In tensor networks, a tensor is represented as a box with open legs, which are indices of this tensor. For example, $\rho_{ABC}=\sum_{i,j,k,i',j',k'}\rho_{ijk,i'j'k'}|ijk\rangle\langle i'j'k'|$ is represented as a box with six legs, three of which are left,  representing row indices $i,j,k$, and the other three are right, representing column indices $i',j',k'$. Open legs represent the noncontracting indices, so the connection of legs is the contraction of indices. For example, tensor $AB$ is graphically represented by connecting the right legs of box $A$ and the left legs of box $B$. $\tr(A)=\sum_i A_{i,i}$ is the contraction of the row and column index of $A$, which can be represented by connecting the left and right legs of box $A$. In addition, tensor product is the operation that does not contract indices. So, to represent $A\otimes B$, we just put boxes $A$ and $B$ together. 

\begin{figure}[htbp]
    \centering
    \includegraphics[scale=0.32]{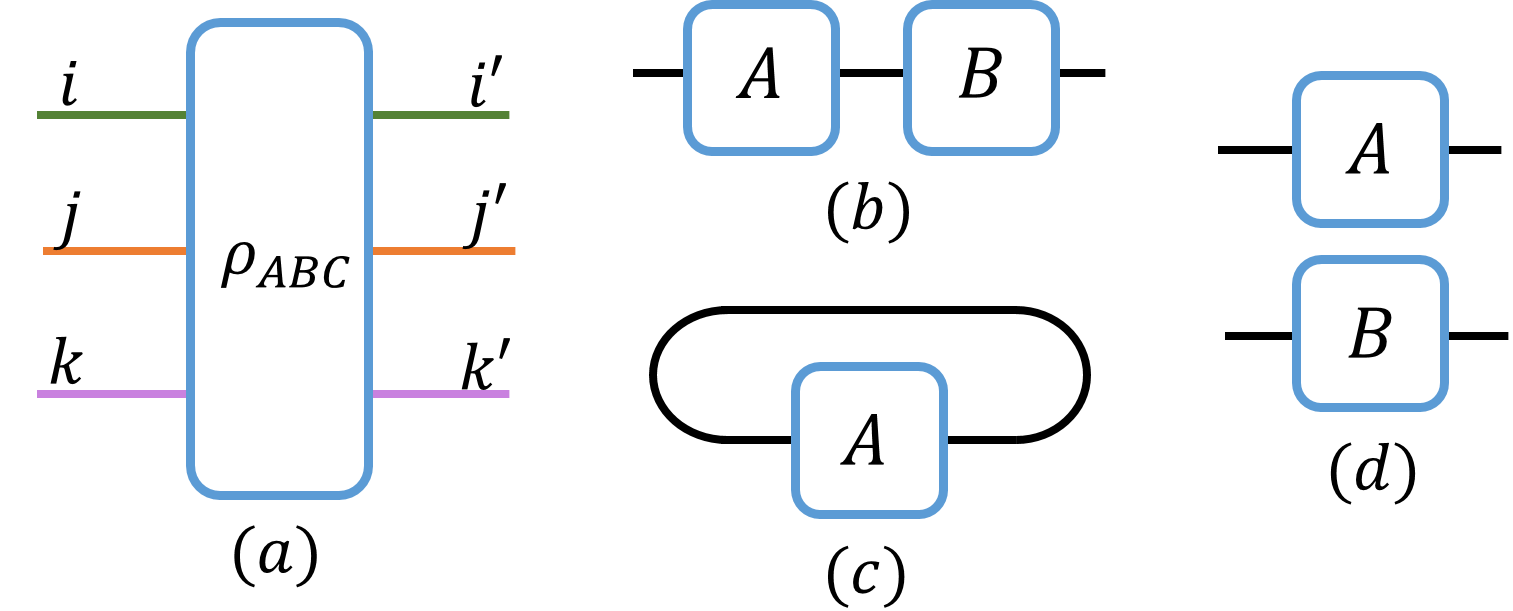}
    \caption{(a) $\rho_{ABC}=\sum_{i,j,k,i',j',k'}\rho_{ijk,i'j'k'}|ijk\rangle\langle i'j'k'|$. (b) $AB$. (c) $\tr(A)$. (d) $A\otimes B$}.
    \label{fig:tensor1}
\end{figure}

The unnormalized maximally entangled state (UMES) and {\footnotesize SWAP} operator are two commonly used operators. UMES $\sqrt{d}|\Psi^+\rangle=\sum_{i=1}|ii\rangle$ is represented as a semicircle with both ends to the left, where $d$ is the dimension of Hilbert space. Many operations can be represented using UMES. Take $A=\sum_{i,j}A_{i,j}|i\rangle\langle j|$ as an example:
\begin{equation}
\begin{aligned}
d\langle \Psi^+|A|\Psi^+\rangle=\sum_k\langle kk|\sum_{i,j}A_{i,j}|i\rangle\langle j|\sum_l|ll\rangle=\sum_{i,j}A_{i,j}\sum_{k,l}\langle k|i\rangle\langle j|l\rangle|k\rangle\langle l|=\sum_{i,j}A_{i,j}|j\rangle\langle i|=A^T
\end{aligned}
\end{equation}
and 
\begin{equation}
\begin{aligned}
d\langle \Psi^+|A\otimes I|\Psi^+\rangle=\sum_k\langle kk|\sum_{i,j}A_{i,j}|i\rangle\langle j|\otimes I\sum_l|ll\rangle=\sum_{i,j}A_{i,j}\sum_{k,l}\langle k|i\rangle\langle j|l\rangle\langle k|l\rangle=\sum_iA_{i,i}=\tr(A).
\end{aligned}
\end{equation}
$\tr(A)$ has been shown in Fig.~\ref{fig:tensor1}(c). {\footnotesize SWAP} $S$, exchanging two indices, can be graphically represented with two crossed curved lines. Using this representation, one can easily prove that $\tr[S(A\otimes B)]=\tr(AB)$, as shown in Fig.~\ref{fig:tensor2}.

\begin{figure}[htbp]
    \centering
    \includegraphics[scale=0.32]{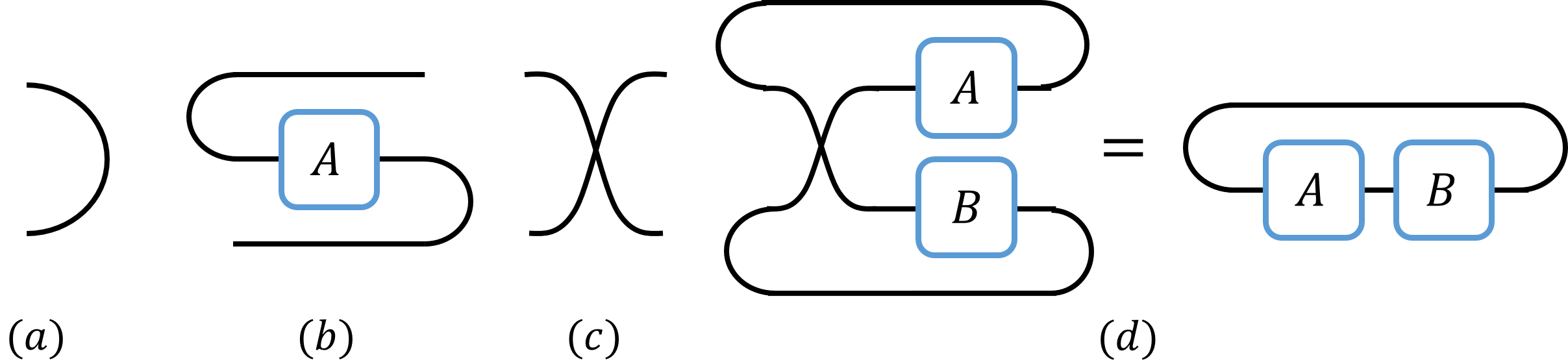}
    \caption{(a) $\sqrt{d}|\Psi^+\rangle=\sum_{i=1}|ii\rangle$. (b) $A^T$. (c) $S$. (d) $\tr[S(A\otimes B)]=\tr(AB)$.}
    \label{fig:tensor2}
\end{figure}

\section{The property of total correlation measure defined in Eq.~\eqref{eq:correlation def}}\label{app:discussion property}

Recall that we take $C(\rho)=-\log_2\mathcal{F}\left(\rho,\bigotimes_{i=1}^k\rho_i\right)$, and the fidelity measure could be 
\begin{equation}\label{eq:F2correlation1}
\begin{aligned}
\mathcal{F}_{\max}\left(\rho,\sigma\right)=\frac{\tr\left(\rho \sigma\right)}{\max \{ \tr(\rho^2),\tr(\sigma^2)\}}
\end{aligned}
\end{equation}
and 
\begin{equation}\label{eq:appF2correlation2}
\begin{aligned}
\mathcal{F}_{\mathrm{GM}}\left(\rho,\sigma\right)=\frac{\tr\left(\rho \sigma\right)}{\sqrt{\tr(\rho^2)\tr(\sigma^2)}},
\end{aligned}
\end{equation}
which is the one used in the main text. Note that the difference is that the denominator is either maximization or geometric mean of the purity, and one can also choose other fidelity measures \cite{Liang_2019}. We take  $C(\rho)=-\log_2\mathcal{F}_{\text{GM}}\left(\rho,\bigotimes_{i=1}^k\rho_i\right)$ as an example to discuss its property. The discussion of $C(\rho)=-\log_2\mathcal{F}_{\text{max}}\left(\rho,\bigotimes_{i=1}^k\rho_i\right)$ is quite similar.

\begin{enumerate}[(1)]
    \item $C(\rho)$ is a faithful total correlation measure, $C(\rho)=0$ iff $\rho = \bigotimes_{i=1}^k\rho_i$, and non-negative for any $\rho$. This property directly follows from the property of fidelity: $\mathcal{F}_{\text{GM}}(\rho,\sigma)<1$, $\forall \rho,\sigma$ and $\mathcal{F}_{\text{GM}}(\rho,\sigma)=1$ iff $\rho = \sigma$.
    
    \item Adding a new party $\rho_{k+1}$ to a $k$-partite state $\rho$ will not cause the increase of the total correlation in the resulting $(k+1)$-partite tensor state $\rho\otimes\rho_{k+1}$, i.e. $C(\rho\otimes\rho_{k+1})=C(\rho)$. 

    \begin{proof}
    \begin{equation}
    \begin{aligned}
    \mathcal{F}_{\text{GM}}(\rho\otimes\rho_{k+1},\bigotimes_{i=1}^k\rho_i\otimes\rho_{k+1})
    &=\frac{\tr\left[(\rho\otimes\rho_{k+1})(\bigotimes_{i=1}^k\rho_i\otimes\rho_{k+1})\right]}{\sqrt{\tr\left[(\rho\otimes\rho_{k+1})^2\right]\tr\left[(\bigotimes_{i=1}^k\rho_i\otimes\rho_{k+1})^2\right]}}\\
    &=\frac{\tr\left(\rho\bigotimes_{i=1}^k\rho_i\right)}{\sqrt{\tr\rho^2\prod_{i=1}^k\tr\rho_i^2}}\frac{\tr\rho_{k+1}^2}{\tr\rho_{k+1}^2}\\
    &=\mathcal{F}_{\text{GM}}(\rho,\bigotimes_{i=1}^k\rho_i).
    \end{aligned}
    \end{equation}
    Hence
    \begin{equation}
    \begin{aligned}
    C(\rho\otimes\rho_{k+1})
    =-\log_2\mathcal{F}_{\text{GM}}(\rho\otimes\rho_{k+1},\bigotimes_{i=1}^k\rho_i\otimes\rho_{k+1})
    =-\log_2 \mathcal{F}_{\text{GM}}(\rho,\bigotimes_{i=1}^k\rho_i)
    =C(\rho).
    \end{aligned}
    \end{equation}
    \end{proof}
    
    \item $C(\rho)$ is not changed under local unitary transformation, i.e. $C(\bigotimes_{i=1}^k u_i\rho\bigotimes_{i=1}^k u_i^\dagger)=C(\rho)$.
    
    \item Local quantum channel, i.e., $\bigotimes_{i=1}^k\Lambda_i$, cannot create total correlation in uncorrelated $k$-partite state $\rho=\bigotimes_{i=1}^k\rho_i$.
    \begin{proof}
    Acting the local channel on the uncorrelated state, the resulting state is still uncorrelated:
    \begin{equation}
    \begin{aligned}
    \bigotimes_{i=1}^k\Lambda_i\left(\bigotimes_{i=1}^k\rho_i\right)=\bigotimes_{i=1}^k\Lambda_i\left(\rho_i\right).
    \end{aligned}
    \end{equation}
    Therefore, the faithfulness of $\mathcal{F}_{\text{GM}}(\cdot)$ directly leads to the invariance of total correlation:
    \begin{equation}
    \begin{aligned}
    C\left[\bigotimes_{i=1}^k\Lambda_i\left(\bigotimes_{i=1}^k\rho_i\right)\right]=C\left(\bigotimes_{i=1}^k\Lambda_i\left(\rho_i\right)\right)=C\left(\bigotimes_{i=1}^k\rho_i\right)=0.
    \end{aligned}
    \end{equation}
    
    \end{proof}
    
    \item $C(\rho)$ is additive under the tensor product. The $(k+l)$-partite total correlation $C(\rho\otimes\sigma)$ equals the sum over of the $k$-partite total correlation $C(\rho)$ and l-partite total correlation $C(\sigma)$.
    \begin{proof}
    \begin{equation}
    \begin{aligned}
    \mathcal{F}_{\text{GM}}(\rho\otimes\sigma,\bigotimes_{i=1}^k\rho_i\otimes\bigotimes_{i=1}^l\sigma_i)
    &=\frac{\tr\left[(\rho\otimes\sigma)(\bigotimes_{i=1}^k\rho_i\otimes\bigotimes_{i=1}^l\sigma_i)\right]}{\sqrt{\tr\left[(\rho\otimes\sigma)^2\right]\tr\left[(\bigotimes_{i=1}^k\rho_i\otimes\bigotimes_{i=1}^l\sigma_i)^2\right]}}\\
    &=\frac{\tr\left(\rho\bigotimes_{i=1}^k\rho_i\right)}{\sqrt{\tr\rho^2\prod_{i=1}^k\tr\rho_i^2}}\frac{\tr\left(\sigma\bigotimes_{i=1}^l\sigma_i\right)}{\sqrt{\tr\sigma^2\prod_{i=1}^l\tr\sigma_i^2}}\\
    &=\mathcal{F}_{\text{GM}}(\rho,\bigotimes_{i=1}^k\rho_i)\mathcal{F}_{\text{GM}}(\sigma,\bigotimes_{i=1}^l\sigma_i).
    \end{aligned}
    \end{equation}
    Hence, 
    \begin{equation}
    \begin{aligned}
    C(\rho\otimes\sigma)&=-\log_2\mathcal{F}_{\text{GM}}(\rho\otimes\sigma,\bigotimes_{i=1}^k\rho_i\otimes\bigotimes_{i=1}^l\sigma_i)\\
    &=-\log_2\mathcal{F}_{\text{GM}}\left(\rho,\bigotimes_{i=1}^k\rho_i\right)-\log\mathcal{F}_{\text{GM}}\left(\sigma,\bigotimes_{i=1}^l\sigma_i\right)\\
    &=C(\rho)+C(\sigma).
    \end{aligned}
    \end{equation}
    \end{proof}
    
\end{enumerate}

It is easy to prove that the other definition of total correlation, $C(\rho)=-\log_2\mathcal{F}_{\text{max}}\left(\rho,\bigotimes_{i=1}^k\rho_i\right)$ satisfies properties (1), (2), (3) and (4), while it fails to meet the additivity condition (5).

Now we discuss how to generalize such total correlation  measure to genuine multipartite correlation measure. In \cite{Bennett_2011}, the authors proposed three postulates that every multipartite genuine correlation measure and indicator should satisfy. They also gave a definition of multipartite genuine correlation based on those postulates: $k$-partite state $\rho$ has genuine $k$-partite correlation if it is nonproduct for any bipartition. Following this definition, we define the $k$-partite genuine correlation measure
\begin{equation}
\begin{aligned}
C_{\text{g.c.}}(\rho)=\min_{A\subset [k]}\left\{-\log\mathcal{F}(\rho,\rho_A\otimes\rho_{\bar{A}})\right\},
\end{aligned}
\end{equation}
where $\mathcal{F}(\rho,\sigma)$ can also be $\mathcal{F}_{\max}(\rho,\sigma)$ or $\mathcal{F}_{\text{GM}}(\rho,\sigma)$. Because of the faithfulness of these two fidelities, $C_{\text{g.c.}}(\rho)=0$ iff $\rho$ is the product in some bipartition $\{A,\bar{A}\}$ of the $k$-partite system. Hence, $C_{\text{g.c}}(\rho)$ satisfies those three postulates proposed in \cite{Bennett_2011}.

\section{Bipartite entanglement criterion using $T_2$}\label{app:T2EW}
In \cite{zhang2008entanglement}, the authors proposed an entanglement criterion which is strictly stronger than the well-known computable cross norm criterion \cite{2002AKai} and dV criterion (the one based on correlation tensor of states) \cite{2007Further},
\begin{equation}\label{eq:enhanced EW}
\begin{aligned}
\norm{\mathcal{R}(\rho_{AB}-\rho_A\otimes\rho_B)}>\sqrt{(1-\tr\rho_A^2)(1-\tr\rho_B^2)},
\end{aligned}
\end{equation}
where $\norm{\cdot}$ denotes the trace norm, and $\mathcal{R}(\cdot)$ represents the realignment operation, $\mathcal{R}(O)_{ij,kl}=O_{ik,jl}$. However, $\norm{\mathcal{R}(\rho_{AB}-\rho_A\otimes\rho_B)}$ is hard to estimate by direct measurements. Hence, based on this criterion, we further construct a new measurable entanglement criterion as follows.
\begin{proposition}\label{prop:EW}
For any separable state $\rho_{AB}$, it should satisfy
\begin{equation}
\begin{aligned}
\tr\rho_{AB}^2 + \tr\rho_A^2 +\tr\rho_B^2 - 2\tr[\rho_{AB}(\rho_A\otimes\rho_B)]-1\le 0,
\end{aligned}
\end{equation}
and the violation indicates the presence of entanglement. 
\end{proposition}
The proof is left to Appendix \ref{proof:EW}. We name this criterion as \emph{$T_2$ separability criterion}. Although weaker than Eq.~\eqref{eq:enhanced EW}, the $T_2$ criterion is equivalent to the well-known \renyi entropy criterion (i.e., $\tr\rho_{AB}^2 \leq \tr\rho_A^2,\tr\rho_B^2$ for the separable $\rho_{AB}$) on pure states and Bell-diagonal states, and shows stronger detection power for some asymmetric states. 

\begin{equation}
\begin{aligned}
\rho_{AB}=(1-p)|\Psi^+\rangle\langle\Psi^+|+p|0+\rangle\langle 0+|
\end{aligned}
\end{equation}
is the mixture of the Bell state $|\Psi^+\rangle$ and the product state $\ket{0+}$. For such state, $T_2$ criterion indicates entanglement for $p<1$, which is same as positive partial transposition (PPT) criterion, the necessary and sufficient condition for $(2\times 2)$-dimensional quantum states. However, the entropy and $p_3$-PPT criterion \cite{elben2020mixedstate} only detect entanglement as $p < 0.5$ and $p < 0.59$, respectively.
A detailed comparison and discussion are left to Appendix \ref{proof:detection power}. It is worth mentioning that this criterion can be generalized to non-full-separability criterion in the multipartite system \cite{zhang2008entanglement}, and we leave it for future study.

\subsection{Proof of Proposition  \ref{prop:EW}}\label{proof:EW}
For simplicity, denote $\mathcal{R}(\rho_{AB}-\rho_A\otimes\rho_B)$ by $R$, and assume the dimension of $\mathcal{H}_A$ is less than the dimension of $\mathcal{H}_B$ $d_A\le d_B$. Then 
\begin{equation}
\begin{aligned}
\norm{R}=\sum_{i=1}^{d_A^2}\lambda_i,
\end{aligned}
\end{equation}
where $\lambda_i\ge 0$ are the singular values of $R$. Although $\norm{R}$ is hard to directly measure, we find that 
\begin{equation}
\begin{aligned}
\tr(RR^\dagger)=\sum_{i=1}^{d_A^2}\lambda_i^2
\end{aligned}
\end{equation}
can be directly measured, and the value of $\norm{R}$ may be bounded by $\tr(RR^\dagger)$.

\begin{lemma}
$\sum_{i=1}^{d_A^2}\lambda_i^2$ can be represented using the purities of $\rho_A$, $\rho_B$, and $\rho_{AB}$ and $T_2=\tr[\rho_{AB}(\rho_A\otimes\rho_B)]$:
\begin{equation}
\begin{aligned}
\tr(RR^\dagger)=\tr(\rho_{AB}-\rho_A\otimes\rho_B)^2=\tr\rho_{AB}^2+\tr\rho_A^2\tr\rho_B^2-2\tr[\rho_{AB}(\rho_A\otimes\rho_B)].
\end{aligned}
\end{equation}
\end{lemma}

\begin{proof}
For simplicity, denote $\rho_{AB}-\rho_A\otimes\rho_B$ by $O_{AB}$. $O_{AB}$ is Hermitian but generally not positive. Hence, we have
\begin{equation}
\begin{aligned}
R^\dagger=[\mathcal{R}(O_{AB})^*]^T=\mathcal{R}^T(O_{AB}^T),
\end{aligned}
\end{equation}
so that the elements of $R^\dagger$ are the elements of $O_{AB}$:
\begin{equation}
\begin{aligned}
R^\dagger_{ij,kl}=[\mathcal{R}(O_{AB})^T]_{kl,ij}=(O_{AB}^T)_{ki,lj}=(O_{AB})_{lj,ki}.
\end{aligned}
\end{equation}
Now we can represent $\tr(RR^\dagger)$ by the index contraction of $\rho_{AB}-\rho_A\otimes\rho_B$:
\begin{equation}
\begin{aligned}
\tr(RR^\dagger)&=\sum_{i,j,k,l}R_{ij,kl}R^\dagger_{kl,ij}=\sum_{i,j,k,l}(O_{AB})_{ik,jl}(O_{AB})_{jl,ik}\\
&=\tr(O_{AB})^2=\tr(\rho_{AB}-\rho_A\otimes\rho_B)^2=\tr\rho_{AB}^2+\tr\rho_A^2\tr\rho_B^2-2\tr[\rho_{AB}(\rho_A\otimes\rho_B)].
\end{aligned}
\end{equation}

\end{proof}

Suppose we have measured the value of $p_2=\sum_{i=1}^{d_A^2}\lambda_i^2=\tr\rho_{AB}^2+\tr\rho_A^2\tr\rho_B^2-2\tr[\rho_{AB}(\rho_A\otimes\rho_B)]$. Then it can be easily proved that 
\begin{equation}\label{eq:moment measure}
\begin{aligned}
\sqrt{p_2}\leq\sum_{i=1}^{d_A^2}\lambda_i\leq d_A\sqrt{p_2}.
\end{aligned}
\end{equation}
The minimum is achieved when $\lambda_1=\sqrt{p_2}$ and $\lambda_i = 0$ for $2\leq i\leq d_A^2$ and the maximum is achieved when $\lambda_i=\sqrt{p_2}/d_A$ for $1\le i\le d_A^2$. Eq.~\eqref{eq:enhanced EW} tells us that separable $\rho_{AB}$ satisfy $\sum_{i=1}^{d_A^2}\lambda_i^2\le\sqrt{(1-\tr\rho_A^2)(1-\tr\rho_B^2)}$. According to the above equation, separable $\rho_{AB}$ satisfy
\begin{equation}
\begin{aligned}
\sqrt{p_2} \le \sqrt{(1-\tr\rho_A^2)(1-\tr\rho_B^2)} \to \tr\rho_{AB}^2+\tr\rho_A^2+\tr\rho_B^2-2\tr[\rho_{AB}(\rho_A\otimes\rho_B)]-1\le 0.
\end{aligned}
\end{equation}

\subsection{Discussion of the detection power of $T_2$ criterion}\label{proof:detection power}
Here we first prove the equivalence of the $T_2$ criterion and the entropy criterion for pure state $|\varphi_{AB}\rangle$. As is known to all, for a pure state, the entropy criterion is a necessary and sufficient condition for separability. So we only need to prove Proposition \ref{prop:EW} is also necessary and sufficient. Taking the Schmidt decomposition of $|\varphi_{AB}\rangle$,
\begin{equation}
\begin{aligned}
|\varphi_{AB}\rangle=\sum_{i=1}^{d_A}\sqrt{\lambda_i}|ii\rangle,
\end{aligned}
\end{equation}
where $0\le\sqrt{\lambda_i}\le1$ are the Schmidt co-efficients. The purities can be easily calculated as
\begin{equation}
\begin{aligned}
\tr\rho_{AB}^2=1 \ , \ \tr\rho_A^2=\tr\rho_B^2=\sum_{i=1}^{d_A}\lambda_i^2,
\end{aligned}
\end{equation}
and
\begin{equation}
\begin{aligned}
\tr[\rho_{AB}(\rho_A\otimes\rho_B)]&=\langle\varphi_{AB}|\rho_A\otimes\rho_B|\varphi_{AB}\rangle\\
&=\sum_{k,l=1}^{d_A}\sqrt{\lambda_k\lambda_l}\langle kk|\left(\sum_{i,j=1}^{d_A}\lambda_i\lambda_j|ij\rangle\langle ij|\right)|ll\rangle\\
&=\sum_{i,j,k,l=1}^{d_A}\sqrt{\lambda_k\lambda_l}\lambda_i\lambda_j\delta_{ik}\delta_{jk}\delta_{il}\delta_{jl}\\
&=\sum_{i=1}^{d_A}\lambda_i^3.
\end{aligned}
\end{equation}
Then the $T_2$ criterion reduces to
\begin{equation}
\begin{aligned}
\sum_{i=1}^{d_A}\lambda_i^2\le\sum_{i=1}^{d_A}\lambda_i^3,
\end{aligned}
\end{equation}
where $\sqrt{\lambda_i}$ are the Schmidt coefficients of $|\varphi_{AB}\rangle$. The normalization and singular value decomposition requires $\sum_{i=1}^{d_A}\lambda_i=1$ and $0\le\lambda_i\le 1$. So the only solution of the above inequality is $\lambda_1=1$ and $\lambda_i=1$ for $2\le i\le d_A$, which means $|\varphi_{AB}\rangle = |\varphi_A\rangle|\varphi_B\rangle$.

For a Bell-diagonal state with the form
\begin{equation}
\begin{aligned}
\rho_{AB} = \frac{1}{4}\left(I_4+r_x\sigma_x\otimes\sigma_x+r_y\sigma_y\otimes\sigma_y+r_z\sigma_z\otimes\sigma_z\right),
\end{aligned}
\end{equation}
where $I_4$ denotes a $4\times 4$ identity matrix, the reduced density matrix can be easily calculated as
\begin{equation}
\begin{aligned}
\rho_A = \rho_B = \frac{1}{2}I_2.
\end{aligned}
\end{equation}
Thus we have
\begin{equation}
\begin{aligned}
\tr\rho_{AB}^2=\frac{1}{4}(1+r_x^2+r_y^2+r_z^2), \ \ \ \tr\rho_A^2=\tr\rho_B^2=\frac{1}{2}, \ \ \ \tr[\rho_{AB}(\rho_A\otimes\rho_B)]=\frac{1}{4}.
\end{aligned}
\end{equation}
Then both the entropy criterion and the $T_2$ criterion indicate entanglement for 
\begin{equation}
\begin{aligned}
r_x^2+r_y^2+r_z^2\ge 1.
\end{aligned}
\end{equation}

Then we take 
\begin{equation}\label{eq:stateBellPlus}
\begin{aligned}
\rho_{AB}=(1-p)|\Psi^+\rangle\langle\Psi^+|+p|0+\rangle\langle 0+|,
\end{aligned}
\end{equation}
a mixture of the two-qubit maximally entangled state and the tensor product of $|0\rangle\langle 0|$ and $|+\rangle\langle +|$, as an example to demonstrate the detection power of the $T_2$ separability criterion. For comparison, we pick three commonly used separability criteria:
\begin{enumerate}
    \item PPT criterion, $W(p)=-min\{\lambda(\rho_{AB}^{T_B})\}$
    \item entropy criterion, $W(p)=\tr\rho_{AB}^2-\tr\rho_A^2$
    \item $p_3$-PPT criterion, $W(p)=(\tr\rho_{AB}^2)^2-\tr(\rho_{AB}^{T_B})^3$
\end{enumerate}
and $T_2$ criterion, $W(p)=\tr\rho_{AB}^2+\tr\rho_A^2+\tr\rho_B^2-2\tr[\rho_{AB}(\rho_A\otimes\rho_B)]-1$, show them in one diagram. For these four criteria, $W(p)>0$ indicates entanglement and the absolute value of $W(p)$ makes no sense. As shown in Fig.~\ref{fig:EWpower}, for states in Eq.~\eqref{eq:stateBellPlus}, the $T_2$ criterion shows the same detection power as the PPT criterion, the necessary and sufficient separability condition for $2\times 2$ states, and better than the $p_3$-PPT criterion and entropy criterion.
 
\begin{figure}[htbp]
    \centering
    \includegraphics[scale=0.5]{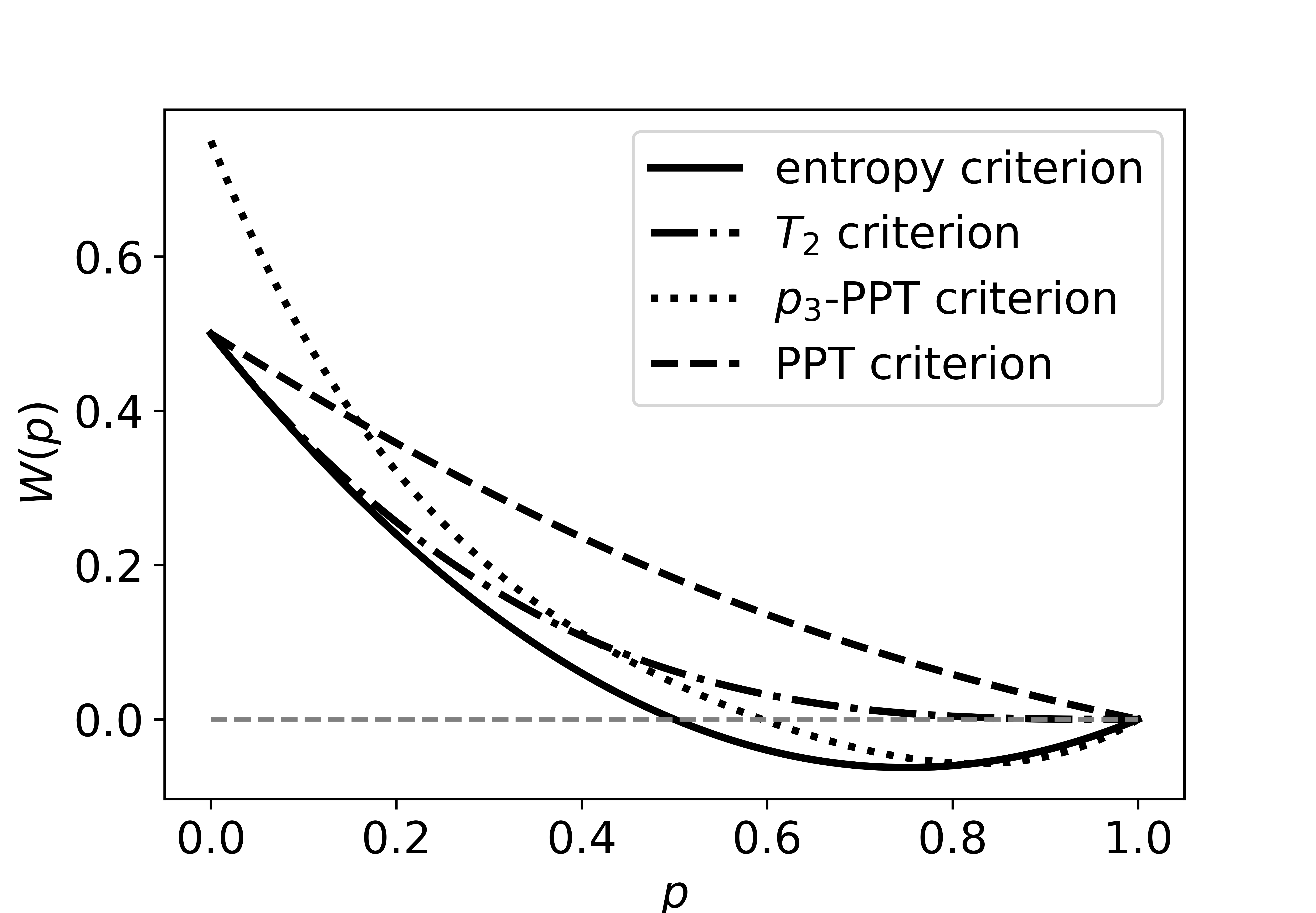}
    \caption{Comparison of the detection power of four separability criteria for states in Eq.~\eqref{eq:stateBellPlus}.}
    \label{fig:EWpower}
\end{figure}

\section{A few Proofs}
\subsection{Proof of Eq.~\eqref{eq:SWAP construction}}
The result in Eq.~\eqref{eq:SWAP construction} was derived in Refs.~\cite{Elben2018Random,Elben2019toolbox}. Here we prove it for completeness. According to the Weingarten integral introduced in Appendix \ref{subsec:random unitary},
\begin{equation}
\begin{aligned}
\Phi_2(X)&=\sum_{\pi,\sigma\in\mathcal{S}_2}C_{\pi,\sigma}\tr(W_\pi X)W_\sigma\\
&=\sum_{\pi,\sigma\in\mathcal{S}_2}\sum_{s,s'}C_{\pi,\sigma}-(-d)^{\delta_{s,s'}}\tr(W_\pi |s,s'\rangle\langle s,s'|)W_\sigma\\
&=\sum_{s,s'}-(-d)^{\delta_{s,s'}}\left\{\left[C_{0,0}\tr(I |s,s'\rangle\langle s,s'|)+C_{1,0}\tr(S |s,s'\rangle\langle s,s'|)\right]I +\left[C_{0,1}\tr(I |s,s'\rangle\langle s,s'|)+C_{1,1}\tr(S |s,s'\rangle\langle s,s'|)\right]S\right\}\\
&=\sum_{s,s'}d(-d)^{\delta_{s,s'}-1}\left\{\left[\frac{1}{d^2-1}-\frac{1}{d(d^2-1)}\delta_{s,s'}\right]I+\left[-\frac{1}{d(d^2-1)}+\frac{1}{d^2-1}\delta_{s,s'}\right]S\right\}\\
&=\sum_{s,s'}\left\{\left[-\frac{1}{d^2-1}+\frac{d}{d^2-1}\delta_{s,s'}\right]I+\left[\frac{1}{d(d^2-1)}+\frac{d^2-d-1}{d(d^2-1)}\delta_{s,s'}\right]S\right\}\\
&=\left\{\left[-\frac{d^2}{d^2-1}+\frac{d^2}{d^2-1}\right]I+\left[\frac{d^2}{d(d^2-1)}+\frac{d^3-d^2-d}{d(d^2-1)}\right]S\right\}\\
&=S
\end{aligned}
\end{equation}
Here in the third line, for simplicity, we use $0$ and $1$ to represent identity and exchange in the subscript of Weingarten matrix element $C_{\pi,\sigma}$. The fifth equal sign is because
\begin{equation}
\begin{aligned}
-(-d)^{\delta_{s,s'}}=-1+(d+1)\delta_{s,s'},
\end{aligned}
\end{equation}
and the sixth equal sign is because
\begin{equation}
\begin{aligned}
\sum_{s,s'}1=d^2 \ \ , \ \ \sum_{s,s'}\delta_{s,s'}=d
\end{aligned}
\end{equation}

\subsection{Proof of Proposition  \ref{prop:Tkglobal}}\label{proof:Tkglobal}

By Born's rule,
\begin{equation}
\begin{aligned}
\mathrm{Pr}(s|U)&=\tr\left[\ket{s_{g_1},\dots,s_{g_k}}\bra{s_{g_1},\dots,s_{g_k}}\left(\bigotimes_{i=1}^kU_{g_i}\right)\rho\left(\bigotimes_{i=1}^kU_{g_i}\right)^\dagger\right]\\
\mathrm{Pr}(s_{g_i}|U_{g_i})&=\tr\left(\ket{s_{g_i}}\bra{s_{g_i}}U_{g_i}\rho_{g_i}U_{g_i}^\dagger\right)
\end{aligned}
\end{equation}
The right-hand side of Eq.~\eqref{eq:core eq correlation} can be written as
\begin{equation}
\begin{aligned}
\mathrm{RHS}&=\tr\left\{\left(\bigotimes_{i=1}^kX_{g_i}\right)\sub{\mbb{E}}{U}\left[\left(\bigotimes_{i=1}^kU_{g_i}\right)^{\otimes 2}\left[\rho\otimes(\bigotimes_{i=1}^k\rho_{g_i})\right]\left(\bigotimes_{i=1}^kU_{g_i}\right) ^{\dagger\otimes 2}\right]\right\}\\
&=\tr\left\{\sub{\mbb{E}}{U}\left[\left(\bigotimes_{i=1}^kU_{g_i}\right)^{\dagger\otimes 2}\left(\bigotimes_{i=1}^kX_{g_i}\right)\left(\bigotimes_{i=1}^kU_{g_i}\right)^{\otimes 2}\right]\left[\rho\otimes(\bigotimes_{i=1}^k\rho_{g_i})\right]\right\}\\
&=\tr\left\{\left[\bigotimes_{i=1}^k\Phi_{g_i}^{\dagger 2}(X_{g_i})\right]\left[\rho\otimes(\bigotimes_{i=1}^k\rho_{g_i})\right]\right\},
\end{aligned}
\end{equation}
where $X_{g_i}$ follows the same definition as the $X$ in Eq.~\eqref{eq:SWAP construction}, and $\Tilde{\Phi}^2(\cdot)=\sub{\mbb{E}}{U}[U^\dagger\cdot U]$. For arbitrary linear operator $A$ and $B$, we have:
\begin{equation}
\begin{aligned}
\tr\left[\Phi^{\dagger k}(A)B\right]=\tr\left[A\Phi^k(B)\right]=\sum_{\pi,\sigma\in\mathcal{S}_k}C_{\pi,\sigma}\tr(W_\pi B)\tr(W_\sigma A)
\end{aligned}
\end{equation}
So, when taking trace, it is always true that
\begin{equation}
\begin{aligned}
\Phi^{\dagger k}(A)=\sum_{\pi,\sigma\in\mathcal{S}_k}C_{\pi,\sigma}\tr(W_\sigma A)W_\pi.
\end{aligned}
\end{equation}
Since $C_{\pi,\sigma}=C_{\sigma,\pi}$ \cite{Gu2013Moments},
\begin{equation}
\begin{aligned}
\Phi^{\dagger k}(A)=\sum_{\pi,\sigma\in\mathcal{S}_k}C_{\pi,\sigma}\tr(W_\pi A)W_\sigma=\Phi^k(A),
\end{aligned}
\end{equation}
so that
\begin{equation}
\begin{aligned}
\mathrm{RHS}&=\tr\left\{\left[\bigotimes_{i=1}^k\Phi_{g_i}^{\dagger 2}(X_{g_i})\right]\left[\rho\otimes(\bigotimes_{i=1}^k\rho_{g_i})\right]\right\}\\
&=\tr\left\{\left(\bigotimes_{i=1}^k S_{g_i}\right)\left[\rho\otimes(\bigotimes_{i=1}^k\rho_{g_i})\right]\right\}\\
&=\tr\left[\rho(\bigotimes_{i=1}^k\rho_{g_i})\right]=T_k.
\end{aligned}
\end{equation}

\subsection{Proof of Proposition  \ref{prop:Tklocal} }\label{proof:Tklocal}
The proof of local protocol is quite similar with global protocol. According to Eq.~\eqref{eq:Xglocal}, define the data processing operator:
\begin{equation}
\begin{aligned}
\tilde{X}_{g_i}=\sum_{\vec{s}_{g_i},\vec{s}_{g_i}'}\tilde{X}_{g_i}(\vec{s}_{g_i},\vec{s}_{g_i}')\ket{\vec{s}_{g_i},\vec{s}_{g_i}'}\bra{\vec{s}_{g_i},\vec{s}_{g_i}'}=\bigotimes_{l\in g_i}X_l.
\end{aligned}
\end{equation}
Substituting this operator into Eq.~\eqref{eq:Tklocal} , the right hand side can be written as
\begin{equation}
\begin{aligned}
\mathrm{RHS} &= \tr\left\{\left(\bigotimes_{i=1}^k\tilde{X}_{g_i}\right)\mathbb{E}_U\left[\left(\bigotimes_{l=1}^nU_l\right)^{\otimes 2}\left[\rho\otimes\left(\bigotimes_{i=1}^k\rho_{g_i}\right)\right]\left(\bigotimes_{l=1}^nU_l\right)^{\dagger\otimes 2}\right]\right\}\\
&=\tr\left\{\mathbb{E}_U\left[\left(\bigotimes_{l=1}^nU_l\right)^{\dagger\otimes 2}\left(\bigotimes_{l=1}^nX_{l}\right)\left(\bigotimes_{l=1}^nU_l\right)^{\otimes 2}\right]\left[\rho\otimes\left(\bigotimes_{i=1}^k\rho_{g_i}\right)\right]\right\}\\
&= \tr\left\{\left(\bigotimes_{l=1}^n\Phi^{\dagger 2}(X_{l})\right)\left[\rho\otimes\left(\bigotimes_{i=1}^k\rho_{g_i}\right)\right]\right\}\\
&=\tr\left\{\left(\bigotimes_{l=1}^nS_l\right)\left[\rho\otimes\left(\bigotimes_{i=1}^k\rho_{g_i}\right)\right]\right\}\\
&=\tr\left[\rho(\bigotimes_{i=1}^k\rho_{g_i})\right]=T_k.
\end{aligned}
\end{equation}

\section{The variance of $\hat{M}$ and $\hat{M}_L$}\label{app:variance1}
\subsection{The variance of $\hat{M}$}
In the following, we figure out the variance $\delta^2 \doteq \mb{Var}\left[\hat{M}(t)\right]$ of the estimator $\hat{M}(t)$ for the $t$-th unitary sampling. The variance of the overall estimator $\hat{M}$ is just $\delta^2/N_U$. Hereafter, for the convenience of our analysis, we take tripartite CRO, $T_3$, as an example to analyze the error scaling and assume that the unitary ensembles are 4-design. The results for $T_3$ can be easily generalized to $T_k$ and the 4-design assumption would not lead to an order of magnitude gap of the leading term. 

With the total variance formula, we have 
\begin{equation} 
\begin{aligned}
\mb{Var}\left[\hat{M}(t)\right] = \sub{\mbb{E}}{U}\left[ \sub{\mbb{E}}{\bm{s}}\left[\hat{M}^2(t)|U\right] \right] - \left[\sub{\mbb{E}}{U}\sub{\mbb{E}}{\bm{s}}\left[\hat{M}(t)|U\right] \right]^2,
\end{aligned}
\end{equation}
where the second term is just $T_3^2$. The first term can be expanded explicitly as
\begin{equation} 
\sub{\mbb{E}}{U}\left[ \sub{\mbb{E}}{\bm{s}}\left[\hat{M}^2(t)|U\right] \right] = \binom{N_M}{4}^{-2} \sum_{\substack{i<j<k<l\\ i'<j'<k'<l'}}  \sub{\mbb{E}}{\bm{s},U}\tr[Q \hat{r}_U(i)\otimes \hat{r}_U(j)\otimes \hat{r}_U(k)\otimes \hat{r}_U(l)] \tr[Q \hat{r}_U(i')\otimes \hat{r}_U(j')\otimes \hat{r}_U(k')\otimes \hat{r}_U(l')].
\end{equation}
To evaluate the above equation, we calculate the terms in the summation depending on the coincidence of the indices, as they label random variables. In Appendix \ref{app:variance2}, we will show that when $D\gg N_M\gg 1$, which is the case of interest, the dominant term of the variance is determined by the case when all the eight indices are coincidental to four indices. This kind of phenomenon is also manifested by the previous theoretical and numerical analyses \cite{elben2020mixedstate,singlezhou,garcia2021quantum} for other quantities based on randomized measurements.
\begin{equation} \label{eq:Var4collision}
\begin{aligned}
\Gamma_4:= & \binom{N_M}{4}^{-2} \sum_{\substack{i=i'<j=j'\\ <k=k'<l=l'}}  \sub{\mbb{E}}{\bm{s},U} \tr[Q \hat{r}_U(i) \otimes \hat{r}_U(j)\otimes \hat{r}_U(k) \otimes \hat{r}_U(l)] \tr[Q \hat{r}_U(i') \otimes \hat{r}_U(j')\otimes \hat{r}_U(k') \otimes \hat{r}_U(l')] \\
&= \binom{N_M}{4}^{-1}  \sub{\mbb{E}}{\bm{s},U}\tr[Q^2 \hat{r}_U(i)\otimes \hat{r}_U(j)\otimes \hat{r}_U(k)\otimes \hat{r}_U(l)] \\
&= \binom{N_M}{4}^{-1}  \tr[Q^2 \Phi_A^4\otimes \Phi_B^4 \otimes \Phi_C^4(\rho_{ABC}^{\otimes 4})] \\
&= \binom{N_M}{4}^{-1}  \tr[\Phi_A^4\otimes \Phi_B^4 \otimes \Phi_C^4(Q^2) \rho_{ABC}^{\otimes 4}] \\
&= \binom{N_M}{4}^{-1}  \tr[\Phi_A^{(1,2)}\otimes \Phi_B^{(1,3)} \otimes \Phi_C^{(1,4)}(Q^2) \rho_{ABC}^{\otimes 4}].
\end{aligned}
\end{equation}
Here, $\Phi^{(i,j)}_A(\cdot):= \sub{\mbb{E}}{U_A}( U_A^{(i)}\otimes U_A^{(j)} ) \cdot ( U_A^{(i)}\otimes U_A^{(j)} )^\dag $ denotes the twofold twirling on subsystems $A$ of the $i$-th and $j$-th copies, similar for $\Phi^{(i,j)}_B(\cdot)$ and $\Phi^{(i,j)}_C(\cdot)$. The last equation of Eq.~\eqref{eq:Var4collision} holds because the observable $Q$ only has nontrivial definitions on the systems $A_1$, $A_2$, $B_1$, $B_3$, $C_1$ and $C_4$ as follows
\begin{equation}
    Q = \left( X_A^{(1,2)}\otimes I_A^{(3,4)} \right) \otimes \left( X_B^{(1,3)}\otimes I_B^{(2,4)} \right) \otimes \left( X_C^{(1,4)}\otimes I_C^{(2,3)} \right).
\end{equation}
Using the Weingarten integral~\cite{Gu2013Moments}, we can get 
\begin{equation}
\begin{aligned}
    \Phi_A^{(1,2)}[(X_A^{(1,2)})^2]=d_A I_A^{(1,2)} + (d_A - 1) S_A^{(1,2)},
\end{aligned}
\end{equation}
thus $\Gamma_4$ shows
\begin{equation}\label{eq:gamma4G}
\begin{aligned}
    \Gamma_4 :&= \binom{N_M}{4}^{-1}  \tr\left[ \left(d_A I_A + (d_A - 1) S^{(1,2)}_A \right)\left(d_B I_B + (d_B - 1) S^{(1,3)}_B \right)\left(d_C I_C + (d_C - 1) S^{(1,3)}_C \right) \rho_{ABC}^{\otimes 4} \right] \\
\end{aligned}
\end{equation}

It is clear that $\Gamma_4$ depends on the input state $\rho_{ABC}$. By expanding Eq.~\eqref{eq:gamma4G}, one gets a few functions of $\rho_{ABC}$, with the coefficients almost $D$. For example, 
\begin{equation}
\begin{aligned}
    &\tr\left[(d_A I_Ad_B I_B d_C I_C) \rho_{ABC}^{\otimes 4} \right]=d_Ad_Bd_C=D, \\
    &\tr\left[ \left((d_A - 1) S^{(1,2)}_A \otimes (d_B - 1) S^{(1,3)}_B \otimes (d_C - 1) S^{(1,3)}_C \right) \rho_{ABC}^{\otimes 4} \right]=(d_A - 1)(d_B - 1)(d_C-1)T_3\leq D.
\end{aligned}
\end{equation}
The term $\Gamma_4$ scales up linearly with $D$. As a result, the variance $\delta^2 \sim \Theta(D)$. 

Here, we take $T_3$ as an example to demonstrate our results. In fact, this conclusion can easily be generalized to $T_k$ measurement for any value of $k$. Following the similar thought in Appendix.~\ref{app:variance2}, we believe that the leading term is also the one which has $(k+1)$ pairs of the same indices, like $\Gamma_4$ in $T_3$ measurement. So the dominant term of variance when measuring $T_k$ is 
\begin{equation}
\begin{aligned}
\Gamma_{k+1}=\binom{N_M}{k+1}^{-1}\tr\left[\bigotimes_{i=1}^k\left(d_iI_i+(d_i-1)S_i^{(1,i+1)}\right)\rho^{\otimes k+1}\right]\sim \Theta(D/N_M^{k+1})
\end{aligned}
\end{equation}

\subsection{The unbiased estimator $\hat{M}_L$ and its variance}
In the above derivation, we consider the variance calculation when the random unitaries $U = U_A\otimes U_B\otimes U_C$ are chosen such that $U_A$, $U_B$, and $U_C$ are elements of unitary 2-design on the corresponding Hilbert spaces. Now we consider the variance estimation in the local strategy mentioned in Sec.~\ref{sec:total correlation}, when the  subsystems $A$, $B$, and $C$ are composed of qubits. In this case, the random unitary twirling is performed locally on each qubit. From Eq.~\eqref{eq:Tklocal}, we can express the tripartite correlation as follows:
\begin{equation} \label{eq:qubit measurement2}
\begin{aligned}
    T_3 = \sum_{ \vec{\bm{a}}, \vec{\bm{b}}, \vec{\bm{c}} }  \tilde{X}^{(1,2)}_A(\vec{a}^1,\vec{a}^2) \tilde{X}^{(1,3)}_B(\vec{b}^1,\vec{b}^3) \tilde{X}^{(1,4)}_C(\vec{c}^1,\vec{c}^4) \prod_{i=1}^4\sub{\mbb{E}}{U}\mathrm{Pr}(\vec{a}^i, \vec{b}^i, \vec{c}^i|U_A, U_B, U_C),
\end{aligned}
\end{equation}
where $\vec{\bm{a}}=(\vec{a}^1,\vec{a}^2,\vec{a}^3,\vec{a}^4)$ denotes the measurement result, being a string whose elements are $n_A$-bit vectors, similar for $\vec{\bm{b}}$ and $\vec{\bm{c}}$. $\tilde{X}_A^{(1,2)}(\vec{a}^1,\vec{a}^2)$ is a function on  $\vec{a}^1$ and $\vec{a}^2$,
\begin{equation}
    \tilde{X}_A^{(1,2)}(\vec{a}^1,\vec{a}^2):=\prod_{i=1}^{n_A} X_{A_i}^{(1,2)}(a^1_i,a^2_i)= 2^{n_A} (-2)^{-D[\vec{a}^1,\vec{a}^2]}.
\end{equation}
We also denote $\tilde{X}_A^{(1,2)}$ as the observable on $\mc{H}^1_A\otimes\mc{H}^2_A$,
\begin{equation}
    \tilde{X}_A^{1,2} = \sum_{ \vec{a}^1,\vec{a}^2 } \tilde{X}_A^{(1,2)}(\vec{a}^1,\vec{a}^2) \ket{\vec{a}^1,\vec{a}^2}\bra{\vec{a}^1,\vec{a}^2}.
\end{equation}
similar for $\tilde{X}_B^{(1,3)}$ and $\tilde{X}_C^{(1,4)}$. Similar to Eq.~\eqref{eq:Mt}, we can define the unbiased estimator based on the local random unitary scheme, 
\begin{equation} \label{eq:MLt}
\hat{M}_L(t) = \binom{N_M}{4}^{-1} \sum_{1\leq i<j<k<l \leq N_M} \tr\left[ \tilde{Q} \left[\hat{r}_U(i) \otimes \hat{r}_U(j) \otimes \hat{r}_U(k) \otimes \hat{r}_U(l)\right] \right],
\end{equation}
where 
\begin{equation}
    \tilde{Q} := \left( \tilde{X}_A^{(1,2)}\otimes I_A^{(3,4)} \right) \otimes \left( \tilde{X}_B^{(1,3)}\otimes I_B^{(2,4)} \right) \otimes \left( \tilde{X}_C^{(1,4)}\otimes I_C^{(2,3)} \right),
\end{equation}
is an observable on $\mc{H}^{\otimes 4}$.

Following the same deduction, to calculate the variance of $\hat{M}_L(t)$, we evaluate the leading term with four coincidences
\begin{equation} \label{eq:Var4collisionLocal}
\begin{aligned}
\tilde{\Gamma}_4:= & \binom{N_M}{4}^{-2} \sum_{\substack{i=i'<j=j'\\ <k=k'<l=l'}}  \sub{\mbb{E}}{\bm{s},U} \tr[\tilde{Q} \hat{r}_U(i) \otimes \hat{r}_U(j)\otimes \hat{r}_U(k) \otimes \hat{r}_U(l)] \tr[\tilde{Q} \hat{r}_U(i') \otimes \hat{r}_U(j')\otimes \hat{r}_U(k') \otimes \hat{r}_U(l')] \\
&= \binom{N_M}{4}^{-1}  \sub{\mbb{E}}{\bm{s},U}\tr[\tilde{Q}^2 \hat{r}_U(i)\otimes \hat{r}_U(j)\otimes \hat{r}_U(k)\otimes \hat{r}_U(l)] \\
&= \binom{N_M}{4}^{-1}  \tr[\tilde{\Phi}_A^4\otimes \tilde{\Phi}_B^4 \otimes \tilde{\Phi}_C^4(\tilde{Q}^2) \rho_{ABC}^{\otimes 4}] \\
&= \binom{N_M}{4}^{-1}  \tr[\tilde{\Phi}_A^{(1,2)}\otimes \tilde{\Phi}_B^{(1,3)} \otimes \tilde{\Phi}_C^{(1,4)}(\tilde{Q}^2) \rho_{ABC}^{\otimes 4}].
\end{aligned}
\end{equation}
The final line is because $\tilde{Q}$ acts nontrivially on the systems $A_1$, $A_2$, $B_1$, $B_3$, $C_1$ and $C_4$. The only difference compared to Eq.~\eqref{eq:Var4collision} is that both the twirling channels $\tilde{\Phi}_A^{(1,2)}$ and $\tilde{Q}_A^{(1,2)}$ have the tensor-product structure on qubits,
\begin{equation}
\begin{aligned}
    \tilde{\Phi}_A^{(1,2)}[(\tilde{Q}_A^{(1,2)})^2]=\bigotimes_{i\in A}\Phi_i^{(1,2)}[(X_i^{(1,2)})^2]=
    \bigotimes_{i\in A}(2 I_i +S^{(1,2)}_i),
\end{aligned}
\end{equation}
similar for operators on $B$ and $C$, and thus $\tilde{\Gamma}_4$ shows
\begin{equation}\label{eq:gamma4L}
\begin{aligned}
    \tilde{\Gamma}_4 :&= \binom{N_M}{4}^{-1}  \tr\left[ \bigotimes_{i\in A}(2 I_i +S^{(1,2)}_i)\bigotimes_{j\in B}(2 I_j +S^{(1,3)}_j)\bigotimes_{k\in C}(2 I_k +S^{(1,4)}_k) \rho_{ABC}^{\otimes 4} \right] \\
    &=\binom{N_M}{4}^{-1} \sum_{A'\subseteq A, B'\subseteq B, C'\subseteq C} 2^{|A|-|A'|}2^{|B|-|B'|}2^{|C|-|C'|}\tr\left[\rho_{A'B'C'}(\rho_{A'}\otimes\rho_{B'}\otimes\rho_{C'})\right]\\
    &\leq \binom{N_M}{4}^{-1} \sum_{A'\subseteq A, B'\subseteq B, C'\subseteq C} 2^n2^{-|A'|}2^{-|B'|}2^{-|C'|}\\
    &= \binom{N_M}{4}^{-1} \sum_{|A'|,|B'|,| C'|}\binom{|A|}{|A'|}\binom{|B|}{|B'|}\binom{|C|}{|C'|}   2^n2^{-|A'|}2^{-|B'|}2^{-|C'|}= \binom{N_M}{4}^{-1}2^n(1+\frac{1}{2})^{n_A+n_B+n_C}=\binom{N_M}{4}^{-1}3^n.\\
    \end{aligned}
\end{equation}
Here in the second line, we expand the terms and the summation of $A'$ runs for all subsets of $A$ including the null set. For example, if $A=\emptyset$, $\tr\left[\rho_{A'B'C'}(\rho_{A'}\otimes\rho_{B'}\otimes\rho_{C'})\right]=\tr\left[\rho_{B'C'}(\rho_{B'}\otimes\rho_{C'})\right]$. The inequality is due to the overlap is less than 1. As a result, the term $\tilde{\Gamma}_4$ is upper bounded by $\binom{N_M}{4}^{-1}3^n=\binom{N_M}{4}^{-1}D^{\log_23}\approx \binom{N_M}{4}^{-1}D^{1.585}$,  and the the variance $\delta^2 \sim O(D^{1.585})$. 

Similarly, the leading term of variance when measuring $T_k$ scales like
\begin{equation}
\begin{aligned}
\tilde{\Gamma}_{k+1}\le \binom{N_M}{k+1}^{-1} 3^n\sim \frac{3^n}{N_M^{k+1}}.
\end{aligned}
\end{equation}

\section{Detailed statistical analysis} \label{app:variance2}

Here, we provide a detailed statistical analysis of the estimation of the tripartite total correlation $\tr[\rho_{ABC}(\rho_A\otimes \rho_B \otimes \rho_C)]$. For simplicity, we will consider the case when $d_A=d_B=d_C=d$. Then $D:=d_A d_B d_C = d^3$.

Recall that we construct an estimator of $T_3$ using these variables in Eq.~\eqref{eq:Mt},
\begin{equation} \label{eq:MtApp}
\hat{M}(t) = \binom{N_M}{4}^{-1} \sum_{1\leq i<j<k<l \leq N_M} \tr\left\{ Q \left[\hat{r}_U(i) \otimes \hat{r}_U(j) \otimes \hat{r}_U(k) \otimes \hat{r}_U(l)\right] \right\}.
\end{equation}

Now, we need to calculate the variance $\delta^2$ of it. In the main text, we show that the core issue is to calculate the term,
\begin{equation} 
\sub{\mbb{E}}{U}\left[ \sub{\mbb{E}}{\bm{s}}(\hat{M}_+^2(t)|U) \right] = \binom{N_M}{4}^{-2} \sum_{\substack{i<j<k<l\\ i'<j'<k'<l'}}  \sub{\mbb{E}}{s,U}\tr[Q \hat{r}_U(i)\otimes \hat{r}_U(j)\otimes \hat{r}_U(k)\otimes \hat{r}_U(l)] \tr[Q \hat{r}_U(i')\otimes \hat{r}_U(j')\otimes \hat{r}_U(k')\otimes \hat{r}_U(l')].
\end{equation}

Based on the coincident number of the sample indices $i,i';j,j';k,k';l,l'$, we may classify the terms as follows,
\begin{enumerate}
\item No coincidence, i.e., the eight sample indices are all different. The number of their terms is $$N_8 = \binom{N_M}{8}\binom{8}{0}\binom{8}{4}.$$
We denote the sum of these terms as $\Gamma_8$.
\item One coincidence. In this case, we further classify the terms based on the coincident index:
\begin{enumerate}
	\item The coincident indices are $i$ and $i'$. The number of their terms is $N_7^{(2)} = \binom{N_M}{7}\binom{6}{3}$. We denote the sum of these terms as $\Gamma_7^{(2)}$.
	\item One of the coincident indices is $i$ or $i'$. The number of their terms is
	$$N_7^{(1)} =\binom{N_M}{7}\left[ \binom{6}{3} + 2\binom{2}{2}\binom{4}{1} + 2\binom{3}{3}\binom{3}{0} \right] = \binom{N_M}{7}\left[ \binom{6}{3} + 10 \right].$$ We denote the sum of these terms as $\Gamma_7^{(1)}$.
	\item None of the coincident indices is $i$ or $i'$. The number of their terms is $$N_7^{(0)} =\binom{N_M}{7}\left\{ \left[ \binom{6}{3} - 2\binom{2}{2}\binom{4}{1} \right] + \left[ \binom{6}{3} - 2\binom{3}{3}\binom{3}{0} \right] + 3\binom{6}{3} \right\} = \binom{N_M}{7}\left( 5\binom{6}{3} - 10 \right). $$ We denote the sum of these terms as $\Gamma_7^{(0)}$.
\end{enumerate}
\item Two coincidences. In this case, we also further classify the terms based on the coincident index:
\begin{enumerate}
	\item The coincident indices contain both $i$ and $i'$. The number of their terms is $N_6^{(2)} =\binom{N_M}{6}\binom{5}{1}\binom{4}{2}.$ We denote the sum of these terms as $\Gamma_6^{(2)}$.
	\item The coincident indices contain either $i$ or $i'$. The number of their terms is $$N_6^{(1)} = \binom{N_M}{6}\left[ \binom{4}{1}\binom{4}{2} + 2\binom{3}{1}\binom{2}{2}\binom{2}{0} \right] = \binom{N_M}{6}\left[ \binom{4}{1}\binom{4}{2} + 6 \right].$$ We denote the sum of these terms as $\Gamma_6^{(1)}$.
	\item The coincident indices do not contain $i$ or $i'$. The number of their terms is $$N_6^{(0)} = \binom{N_M}{6}\left\{ \binom{3}{1}\left[\binom{4}{2} - 2\binom{2}{2} \right] + \binom{2}{1}\binom{4}{2} + \binom{1}{1}\binom{4}{2} \right\} = \binom{N_M}{6}\left\{\left[\binom{3}{1} + \binom{2}{1} + \binom{1}{1} \right]\binom{4}{2} -6 \right\}.$$ We denote the sum of these terms as $\Gamma_6^{(0)}$.
\end{enumerate}
\item Three coincidences. In this case, we also further classify the terms based on the coincident index:
\begin{enumerate}
	\item The coincident indices contain both $i$ and $i'$. The number of their terms is $N_5^{(2)} = \binom{N_M}{5}\binom{4}{2}\binom{2}{1}$. We denote the sum of these terms as $\Gamma_5^{(2)}$.
	\item The coincident indices contain either $i$ or $i'$. The number of their terms is $N_5^{(1)} =\binom{N_M}{5}\binom{3}{2}\binom{2}{1}$. We denote the sum of these terms as $\Gamma_5^{(1)}$.
	\item The coincident indices do not contain $i$ or $i'$. The number of their terms is $N_5^{(0)} = \binom{N_M}{5}\binom{2}{2}\binom{2}{1}$. We denote the sum of these terms as $\Gamma_5^{(0)}$.
\end{enumerate}
\item Four coincidences, i.e., the eight sample indices collapse to four degenerated indices. The number of their terms is $N_4 = \binom{N_M}{4}\binom{4}{4}\binom{0}{0}$. We denote the sum of these terms as $\Gamma_4$.
\end{enumerate}

We can then expand the variance term as follows,
\begin{equation} 
\sub{\mbb{E}}{U}\left[ \sub{\mbb{E}}{s}\left[\hat{M}^2(t)|U\right] \right] = \Gamma_8 + \left(\Gamma_7^{(2)}+ \Gamma_7^{(1)} + \Gamma_7^{(0)}\right) + \left(\Gamma_6^{(2)} + \Gamma_6^{(1)} + \Gamma_6^{(0)}\right) + \left(\Gamma_5^{(2)} + \Gamma_5^{(1)} + \Gamma_5^{(0)}\right) + \Gamma_4.
\end{equation}

In what follows, we focus on the case when $D\gg N_M\gg 1$, which is the case of interest. We want to show that the term $\Gamma_4$ owns the highest dependence of the scaling of $D$.

\begin{proposition} \label{prop:Gamma}
When $D\gg N_M\gg 1$, in the tripartite total correlation estimation task, the different variance terms have the following dependence on the dimension $D$:
\begin{equation}
\begin{aligned}
&\Gamma_8 = O(1), \\
&\Gamma_7^{(2)} = O(1),\quad \Gamma_7^{(1)} = \Gamma_7^{(0)} = O(1), \\
&\Gamma_6^{(2)} = O(d) = O(D^{1/3}), \quad \Gamma_6^{(1)} = \Gamma_6^{(0)} = O(1), \\
&\Gamma_5^{(2)} = O(d^2) = O(D^{2/3}), \quad \Gamma_5^{(1)} = \Gamma_5^{(0)} = O(1), \\
&\Gamma_4 = \Theta(d^3) = \Theta(D).
\end{aligned}
\end{equation}
\end{proposition}

\begin{proof}

We will study the variance terms one by one. 
\begin{equation}
\begin{aligned}
&\binom{N_M}{4}^{2} \binom{N_M}{8}^{-1}\binom{8}{4}^{-1}\Gamma_8 =\sub{\mbb{E}}{U} \tr[Q (\rho_U)^{\otimes 4}]^2 =  \sub{\mbb{E}}{U} \tr[Q^{\otimes 2} (\rho_U)^{\otimes 8}] \\
&= \tr\left[\rho^{\otimes 8} \Phi_A^{(1,2,5,6)}\left(X_A^{(1,2)}\otimes X_A^{(5,6)}\right) \otimes \Phi_B^{(1,3,5,7)}\left(X_B^{(1,3)}\otimes X_B^{(5,7)}\right) \otimes \Phi_C^{(1,4,5,8)}\left(X_C^{(1,4)}\otimes X_C^{(5,8)}\right) \right] \\
&= \sum_{\substack{\pi_1,\pi_2,\pi_3,\\ \sigma_1, \sigma_2, \sigma_3 \in S_4}} C_{\pi_1,\sigma_1} C_{\pi_2,\sigma_2} C_{\pi_3,\sigma_3} \tr\left[\rho^{\otimes 8} \left( W^A_{\pi_1} \otimes W^B_{\pi_2} \otimes W^C_{\pi_3} \right) \right] \tr\left[ \left(X_A^{(1,2)}\otimes X_A^{(5,6)}\right) W^A_{\sigma_1} \right] \\
& \times\qquad\qquad\qquad \tr\left[ \left(X_B^{(1,3)}\otimes X_B^{(5,7)}\right) W^B_{\sigma_2} \right] \tr\left[ \left(X_C^{(1,4)}\otimes X_C^{(5,8)}\right) W^C_{\sigma_3} \right] \\
\end{aligned}
\end{equation}
Here, $\rho_U := U \rho U^\dag$. In the third equality, we assume the random unitaries form a unitary $4$-design. In the fourth equality, we use the Weingarten integral formula~\cite{Gu2013Moments}.

The value of $\Gamma_8$ is obviously state dependent. However, when we consider the asymptotic case when $d\gg N_M \gg 1$, to analyze the scaling of $\Gamma_8$ with $d$, we always consider a \emph{pure tensor state} $\rho= \ket{\psi}_A\bra{\psi}\otimes \ket{\psi}_B\bra{\psi}\otimes\ket{\psi}_C\bra{\psi}$. In this case, the values 
\begin{equation}
\tr\left[\rho^{\otimes 8} \left( W^A_{\pi_1} \otimes W^B_{\pi_2} \otimes W^C_{\pi_3} \right) \right]
\end{equation}
are always $1$. We remark that, if $\rho$ is not a pure tensor state, the absolute value of this term is always smaller than $1$. From this perspective, the pure-tensor-state case will always provide an upper bound of the variance term dependence.

If we set the state $\rho$ to be a pure tensor state, then we have
\begin{equation}
\begin{aligned}
&\binom{N_M}{4}^{2} \binom{N_M}{8}^{-1}\binom{8}{4}^{-1}\Gamma_8\\
&= \sum_{\substack{\pi_1,\pi_2,\pi_3,\\ \sigma_1, \sigma_2, \sigma_3 \in S_4}} C_{\pi_1,\sigma_1} C_{\pi_2,\sigma_2} C_{\pi_3,\sigma_3} \tr\left[ \left(X_A^{(1,2)}\otimes X_A^{(5,6)}\right) W^A_{\sigma_1} \right] \tr\left[ \left(X_B^{(1,3)}\otimes X_B^{(5,7)}\right) W^B_{\sigma_2} \right] \tr\left[ \left(X_C^{(1,4)}\otimes X_C^{(5,8)}\right) W^C_{\sigma_3} \right] \\
&= \sum_{\sigma_1, \sigma_2, \sigma_3 \in S_4} \tr\left[ \left(X_A^{(1,2)}\otimes X_A^{(5,6)}\right) W^A_{\sigma_1} \right] \tr\left[ \left(X_B^{(1,3)}\otimes X_B^{(5,7)}\right) W^B_{\sigma_2} \right] \tr\left[ \left(X_C^{(1,4)}\otimes X_C^{(5,8)}\right) W^C_{\sigma_3} \right] \\
&\times \qquad\qquad\qquad \sum_{\pi_1, \pi_2, \pi_3 \in S_4} C_{\pi_1,\sigma_1} C_{\pi_2,\sigma_2} C_{\pi_3,\sigma_3} \\
&= \left\{  \sum_{\sigma_1 \in S_4} \tr\left[ \left(X_A^{(1,2)}\otimes X_A^{(5,6)}\right) W^A_{\sigma_1} \right] \sum_{\pi_1 \in S_4} C_{\pi_1,\sigma_1}  \right\}^3 \\
&= \left\{ \frac{(d-1)!}{(d+3)!} \sum_{\sigma \in S_4} \tr\left[ \left(X_A^{\otimes 2}\right) W^A_{\sigma} \right]  \right\}^3  \\
&= \left\{ \frac{(d-1)!}{(d+3)!} d(d+1)(d^2+9d+2)  \right\}^3.
\end{aligned}
\end{equation}
In the last equation, we have used Proposition  \ref{prop:SumPermu}. Therefore, $\Gamma_8\sim O(1)$.

Following similar methods, if we assume the state to be a pure tensor state, we can prove that
\begin{equation}
\begin{aligned}
\frac{1}{N_7^{(2)}}\binom{N_M}{4}^{2} \Gamma_7^{(2)}
&= \left\{ \frac{(d-1)!}{(d+2)!} \sum_{\sigma \in S_3} \tr\left[ \left(X_A^{(12,13)}\right) W^A_{\sigma} \right]  \right\}^3  \\
&= \left\{ \frac{(d-1)!}{(d+2)!}\; 3d^2(d+1)  \right\}^3 \sim O(1) \\
\end{aligned}
\end{equation}

\begin{equation}
\begin{aligned}
& \frac{1}{N_7^{(1)}}\binom{N_M}{4}^{2} \Gamma_7^{(1)} = \frac{1}{N_7^{(0)}}\binom{N_M}{4}^{2} \Gamma_7^{(0)}\\
&= \left\{ \frac{(d-1)!}{(d+2)!} \sum_{\sigma \in S_3} \tr\left[ \left(X_A^{(12,13)}\right) W^A_{\sigma} \right]  \right\} \left\{ \frac{(d-1)!}{(d+3)!} \sum_{\sigma \in S_4} \tr\left[ \left(X_A^{\otimes 2}\right) W^A_{\sigma} \right]  \right\}^2  \\
&= \left\{ \frac{(d-1)!}{(d+2)!}\; 3d^2(d+1)  \right\} \left\{ \frac{(d-1)!}{(d+3)!} d(d+1)(d^2+9d+2)  \right\}^2 \sim O(1) \\
\end{aligned}
\end{equation}

\begin{equation}
\begin{aligned}
& \frac{1}{N_6^{(2)}}\binom{N_M}{4}^2 \Gamma_6^{(2)}\\
&= \left\{ \frac{(d-1)!}{(d+2)!} \sum_{\sigma \in S_3} \tr\left[ \left(X_A^{(12,13)}\right) W^A_{\sigma} \right]  \right\}^2 \left\{ \frac{(d-1)!}{(d+1)!} \sum_{\sigma \in S_2} \tr\left[ X_A^2 W^A_{\sigma} \right]  \right\}  \\
&= \left\{ \frac{(d-1)!}{(d+2)!}\; 3d^2(d+1)  \right\}^2 \left\{ \frac{(d-1)!}{(d+1)!} d(2d-1)(d+1)  \right\} \sim O(d) =O(D^{1/3})\\
\end{aligned}
\end{equation}

\begin{equation}
\begin{aligned}
& \frac{1}{N_6^{(1)}}\binom{N_M}{4}^{2} \Gamma_6^{(1)} = \frac{1}{N_6^{(0)}}\binom{N_M}{4}^{2} \Gamma_6^{(0)} \\
&= \left\{ \frac{(d-1)!}{(d+2)!} \sum_{\sigma \in S_3} \tr\left[ \left(X_A^{(12,13)}\right) W^A_{\sigma} \right]  \right\}^3  \\
&= \left\{ \frac{(d-1)!}{(d+2)!}\; 3d^2(d+1)  \right\}^3 \sim O(1) \\
\end{aligned}
\end{equation}

\begin{equation}
\begin{aligned}
& \frac{1}{N_5^{(2)}}\binom{N_M}{4}^{2} \Gamma_5^{(2)} \\
&= \left\{ \frac{(d-1)!}{(d+2)!} \sum_{\sigma \in S_3} \tr\left[ \left(X_A^{(12,13)}\right) W^A_{\sigma} \right]  \right\} \left\{ \frac{(d-1)!}{(d+1)!} \sum_{\sigma \in S_2} \tr\left[ X_A^2 W^A_{\sigma} \right]  \right\}^2  \\
&= \left\{ \frac{(d-1)!}{(d+2)!}\; 3d^2(d+1)  \right\} \left\{ \frac{(d-1)!}{(d+1)!} d(2d-1)(d+1)  \right\}^2 \sim O(d^2) = O(D^{2/3})\\
\end{aligned}
\end{equation}

\begin{equation}
\begin{aligned}
& \frac{1}{N_5^{(1)}}\binom{N_M}{4}^{2} \Gamma_5^{(1)}= \frac{1}{N_5^{(0)}}\binom{N_M}{4}^{2} \Gamma_5^{(0)} \\
&= \left\{ \frac{(d-1)!}{(d+2)!} \sum_{\sigma \in S_3} \tr\left[ \left(X_A^{(12,13)}\right) W^A_{\sigma} \right]  \right\}^3  \\
&= \left\{ \frac{(d-1)!}{(d+2)!}\; 3d^2(d+1)  \right\}^3 \sim O(1) \\
\end{aligned}
\end{equation}

The term $\Gamma_4$ has already been calculated in Eq.~\eqref{eq:gamma4G}, which is
\begin{equation}
\begin{aligned}
    \Gamma_4 = \binom{N_M}{4}^{-1}  \tr\left[ \left(d I_A + (d - 1) S^{(1,2)}_A \right)\left(d I_B + (d - 1) S^{(1,3)}_B \right)\left(d I_C + (d - 1) S^{(1,3)}_C \right) \rho_{ABC}^{\otimes 4} \right] \sim \Theta(D).
\end{aligned}
\end{equation}

\end{proof}

In the above proofs, we have used the following results.

\begin{proposition} \label{prop:SumPermu}
For the observables defined on several copies of $\mc{H}_A$,
\begin{equation}
\begin{aligned}
& X^2 = \sum_{ \bm{a}\in \mbb{Z}^2_d } X^2(a^1,a^2) \ket{a^1,a^2}\bra{a^1,a^2}, \\
& X^{(12,13)} = \sum_{ \bm{a}\in \mbb{Z}^3_d } X(a^1,a^2)X(a^1,a^3) \ket{a^1,a^2,a^3}\bra{a^1,a^2,a^3}, \\
& X^{\otimes 2} = \sum_{ \bm{a}\in \mbb{Z}^4_d } X(a^1,a^2)X(a^3,a^4) \ket{a^1,a^2,a^3,a^4}\bra{a^1,a^2,a^3,a^4}, \\
\end{aligned}
\end{equation}
where $X(a^1,a^2) = -(-d)^{\delta_{a^1,a^2}}$. When the dimension of $\mc{H}_A$ is $d$, we have
\begin{equation}
\begin{aligned}
& \sum_{\sigma\in S_2} \tr( W_\sigma X^2 ) = d(2d-1)(d+1), \\
& \sum_{\sigma\in S_3} \tr( W_\sigma X^{(12,13)} ) = 3d^2(d+1), \\
& \sum_{\sigma\in S_4} \tr( W_\sigma X^{\otimes 2} ) = d(d+1)(d^2+9d+2).
\end{aligned}
\end{equation}
\end{proposition}

\begin{proof}
First we note that,
\begin{equation}
\begin{aligned}
\sum_{\sigma\in S_2} &\tr( W_\sigma X^2 ) = \tr[(I+S)X^2] \\
&= \sum_{a^1,a^2} d^2 (d^2)^{\delta[a^1,a^2]-1} \bra{a^1,a^2} (I+S) \ket{a^1,a^2} \\
&= 2d^3 + d^2 - d \sim O(d^3). 
\end{aligned}
\end{equation}

Then we consider the term for $X^{(12,13)}$. Following the analysis in Ref.~\cite{singlezhou}, we denote the cycle structures (conjugate classes) of the elements $\sigma\in S_t$ using the partition numbers $[\xi_1,\xi_2,\dots, \xi_k]$ where $\xi_1\geq \xi_2 \geq \cdots \geq \xi_k\geq 0$. Also, we can classify $t$-dit strings $\bm{a}$ by the partition numbers $\lambda{\bm{a}}$. For example, the partition number $\lambda(\bm{a})$ of $\bm{a} = (1,2,1)$ is $[2,1]$.

After classifying the cycle structure of the elements in $S_t$, for a diagonal observable $Q$ in the Hilbert space $\mc{H}_A^{\otimes t}$, we have
\begin{equation}
\begin{aligned}
\sum_{\pi\in S_t} \tr(W_\pi Q) = \sum_{\bm{a}\in \mbb{Z}_d^t} Q(\bm{a}) \sum_{\pi\in S_t} \bra{\bm{a}} W_\pi \ket{\bm{a}} = \sum_{\bm{a}\in \mbb{Z}^t_d} Q(\bm{a}) T(\bm{a}),
\end{aligned}
\end{equation}
where
\begin{equation}
T(\bm{a}) = \sum_{\pi \in S_t} \tr(W_\pi \ket{\bm{a}}\bra{\bm{a}}) = \prod_{i=1}^k (\lambda_i(\bm{a}))!.
\end{equation}
The value of $T(\bm{a})$ only depends on the cycle structure of $\bm{a}$, i.e., how many values in $\bm{a}$ are the same.

Furthermore, to calculate $\sum_{\bm{a}\in \mbb{Z}_d^t} Q(\bm{a}) T(\bm{a})$, we first classify all the $t$-dit strings $\bm{a}\in \mbb{Z}_d^t$ by their partitions $\lambda{\bm{a}}$, and then futher divide them by the weight of the subsystems. By counting the weight of the subsystems, we define the ``subtypes'' $\{j_\lambda\}$ of a given partition class $\lambda$ of $\bm{a}$. 
The partition $\lambda$ and subtype $j_\lambda$ determine the value of $T(\bm{a})=T_{\lambda(\bm{a})}$ and $Q(\bm{a}) = Q(j_\lambda)$, respectively. We then count the number of elements $\bm{a}$ in all partition classes and subtypes, and finally figure out the results.

To be more specific,
\begin{equation} \label{eq:sumQaTa}
\begin{aligned}
\sum_{\bm{a}\in \mbb{Z}_d^{t}} Q(\bm{a}) T(\bm{a})&= \sum_{\lambda} T_\lambda \sum_{\bm{a}\in \lambda} Q(\bm{a})\\
&= \sum_{\lambda} T_\lambda \sum_{(j_\lambda)\in\lambda}\#\{j_\lambda\} Q(j_\lambda).
\end{aligned}
\end{equation}

For the $X^{(12,13)}$ case, we need to estimate $\sum_{\bm{a}\in \mbb{Z}_d^3} X^{(12,13)}(\bm{a}) T(\bm{a}) $. When $t=3$, the partition class of $\mbb{Z}_d^3$ determines the subsystem weight in $\mbb{Z}_d^3$. We classify the elements by $\lambda$ and list the values of $T_\lambda$ and $X^{(12,13)}(j_\lambda)$ in Table~\ref{tab:classX1213}. 

\begin{table}[htbp]
\begin{tabular}{c|cc||c|ccc}
\hline\hline
Partition classes $\lambda$ & $\#\{\lambda\}$ & $T_\lambda$ & Subtype $j_\lambda: a^1 | a^2 | a^3$ & $\#\{j_\lambda\}$ & $(wt(a^1, a^2), wt(a^1, a^3))$ & $X^{(12,13)}(j_\lambda)$ \\ \hline
$[111]$  & $1\times A_d^3$ & $1$ & $a|b|c$ & $1\times A_d^3$ & $(1,1)$ & $1$ \\ \hline
$[21]$  & $3\times A_d^2$ & $2$ & $a|a|b$ & $2\times A_d^2$ & $(2,1)$ & $-d$ \\ 
 &  &  & $b|a|a$ & $1\times A_d^2$ & $(1,1)$ & $1$ \\ 
\hline
$[3]$  & $1\times A_d^1$ & $6$ & $a|a|a$ & $1\times A_d^1$ & $(2,2)$ & $d^2$ \\ \hline
\hline
\end{tabular}
\caption{The classes and elements number of $\bm{a}$ for $T_\lambda$ and $X^{(12,13)}(j_\lambda)$.} \label{tab:classX1213}
\end{table}
Therefore,
\begin{equation}
\begin{aligned}
\sum_{\pi\in S_3} \tr(W_\pi X^{(12,13)}) &= \sum_{\bm{a}\in \mbb{Z}^3_d} X^{(12,13)} T(\bm{a}) \\
&= 1\times 1\times A_d^3 + (-d)\times 2\times 2A_d^2 + 1\times 2 \times A_d^2 + d^2\times 6\times A_d^1 \\
&= 3d^2(d+1).
\end{aligned}
\end{equation}

For the $X^{\otimes 2}$ case, we need to estimate $\sum_{\bm{a}\in \mbb{Z}_d^4} X^{\otimes 2}(\bm{a}) T(\bm{a}) $. When $t=4$, the partition class of $\mbb{Z}_d^4$ determines the subsystem weight in $\mbb{Z}_d^4$. We classify the elements by $\lambda$ and list the values of $T_\lambda$ and $X^{\otimes 2}(j_\lambda)$ in Table~\ref{tab:classXX}. 

\begin{table}[htbp]
\begin{tabular}{c|cc||c|ccc}
\hline\hline
Partition classes $\lambda$ & $\#\{\lambda\}$ & $T_\lambda$ & Subtype $j_\lambda: a^1 a^2 | a^3 a^4$ & $\#\{j_\lambda\}$ & $(wt(a^1, a^2), wt(a^3, a^4))$ & $X^{\otimes 2}(j_\lambda)$ \\ \hline
$[1111]$  & $1\times A_d^4$ & $1$ & $ab|cd$ & $1\times A_d^4$ & $(1,1)$ & $1$ \\ \hline
$[211]$  & $6\times A_d^3$ & $2$ & $aa|bc$ & $2 \times A_d^3$ & $(2,1)$ & $-d$ \\ 
& & & $ab|ac$ & $4 \times A_d^3$ & $(1,1)$ & $1$ \\\hline
$[22]$  & $3\times A_d^2$ & $4$ & $aa|bb$ & $1 \times A_d^2$ & $(2,2)$ & $d^2$ \\ 
& & & $ab|ab$ & $2 \times A_d^2$ & $(1,1)$ & $1$ \\  \hline
$[31]$  & $4\times A_d^2$ & $6$ & $aa|ab$ & $4\times A_d^2$ & $(2,1)$ & $-d$ \\\hline
$[4]$  & $1\times A_d^1$ & $24$ & $aa|aa$ & $1\times A_d^1$ & $(2,2)$ & $d^2$ \\ \hline
\hline
\end{tabular}
\caption{The classes and elements number of $\bm{a}$ for $T_\lambda$ and $X^{\otimes 2}(j_\lambda)$. The sub-type $j_\lambda$ is determined by the weight of subsystems $a^1a^2$ and $a^3a^4$. $\#\{j_\lambda\}$ denote the number of elements contained in the sub-type $j_\lambda$. } \label{tab:classXX}
\end{table}

Therefore,
\begin{equation}
\begin{aligned}
\sum_{\pi\in S_4} \tr(W_\pi X^{\otimes 2}) &= \sum_{\bm{a}\in \mbb{Z}^4_d} X^{\otimes 2}(\bm{a}) T(\bm{a}) \\
&= 1\times 1\times A_d^4 + (-d)\times 2\times 2A_d^3 + 1\times 2\times 4A_d^3 + d^2\times 4\times A_d^2 + 1\times 4\times 2A_d^2 + (-d)\times 6\times 4A_d^2 + d^2\times 24 \times A_d^1 \\
&= d(d+1)(d^2+9d+2).
\end{aligned}
\end{equation}

\end{proof}

\section{Concurrence estimation}\label{sec:concurrence}

Concurrence was first proposed as a byproduct of entanglement of formation (EF) \cite{wootters2001entanglement}, and it was proved that for a bi-qubit system, quantum  concurrence gives a lower bound for EF. After the proposal of concurrence of bi-qubit systems, many works about how to generalize it to multipartite systems were proposed.   In \cite{Beacom2004Spectroscopy}, the author defined the quantum concurrence of $n$-qubit pure state $\Psi\in \mathcal{H}_2^{\otimes n}$ as:
\begin{equation}
\begin{aligned}
C_n(\Psi)=2^{1-n/2}\sqrt{(2^n-2)-\sum_i\tr\rho_i^2},
\end{aligned}
\end{equation}
which is a natural generalization of two-qubit concurrence,
where $i$ labels $(2^n-2)$ nontrivial subsystems and $\rho_i$ is the corresponding density matrix of it. Then the quantum concurrence of multipartite mixed state can be defined as $C_n(\rho)=\mathrm{inf}\sum_ip_iC_n(\Psi_i)$, where the infimum is taken over all pure-state decomposition of $\rho$, just like the definition of EF. In \cite{Aolita_2006}, the author proved that $C_n(\Psi)$ can be measured using just one factorizable observable acting on two identical copies of $\Psi$:
\begin{equation}\label{eq:con 1}
\begin{aligned}
C_N(\Psi)=\sqrt{\langle \Psi|\otimes\langle \Psi|A|\Psi\rangle\otimes|\Psi\rangle}, \ \ \ A=4(1-P^+_1\otimes\cdots\otimes P^+_N).
\end{aligned}
\end{equation}
where $P^+_i=(I_i+S_i)/2$ is the projector that can project states in $\mathcal{H}_i\otimes\mathcal{H}_i$, to symmetric subsystem $\mathcal{H}_i\odot\mathcal{H}_i$. Following this equation, if one wants to estimate $C_N(\Psi)$, he just needs to prepare two identical copies of $\Psi$ and measure observable $P^+_1\otimes\cdots\otimes P^+_N$ on $\Psi^{\otimes 2}$. However, with the help of 
randomized measurements, $C_N(\Psi)$ can be measured with single copies of $\Psi$.

Referring to the $k$-fold twirling channel acting on $X\in\mathcal{H}_2^{\otimes k}$:
\begin{equation}
\begin{aligned}
\Phi^k(X)=\sum_{\pi,\sigma\in\mathcal{S}_k}C_{\pi,\sigma}\tr(XW_\pi)W_\sigma.
\end{aligned}
\end{equation}
Suppose $X=|\psi\rangle\langle\psi|^{\otimes k}$, $|\psi\rangle\in\mathcal{H}_2$, one can easily prove that $\tr(|\psi\rangle\langle\psi|^{\otimes k}W_\pi)=1,\forall\pi\in\mathcal{S}_k$, so that
\begin{equation}\label{eq:special twirling}
\begin{aligned}
\Phi^k(|\psi\rangle\langle\psi|^{\otimes k})=\sum_{\pi,\sigma\in\mathcal{S}_k}C_{\pi,\sigma}W_\sigma=\frac{(d-1)!}{(d+k-1)!}\sum_{\pi\in\mathcal{S}_k}W_\pi.
\end{aligned}
\end{equation}
The second equal sign is because the sum of one row or one column of Weingarten matrix is constant: 
\begin{equation}
\begin{aligned}
\sum_{\alpha\in\mathcal{S}_k}Wg_{k}^{\mathcal{U}_d}(\alpha,\beta)=(d-1)!/(d+k-1)!
\end{aligned}
\end{equation}
where $d$ is the dimension of random unitary and $|\psi\rangle$ in Eq.~\eqref{eq:special twirling}. According to Eq.~\eqref{eq:special twirling}, one can generate the projector $P^+=(I+S)/2$ by two-fold twirling channel
\begin{equation}\label{eq:con 2}
\begin{aligned}
\Phi^2\left[(|\psi\rangle\langle\psi|)^{\otimes 2}\right]=\sum_{\pi,\sigma\in\mathcal{S}_2}C_{\pi,\sigma}W_\sigma=\frac{1}{6}(S+I)=\frac{1}{3}P^+.
\end{aligned}
\end{equation}
Recall that virtual operations can be constructed via random evolution and data post processing. According to Eqs.~\eqref{eq:con 1} and \eqref{eq:con 2},  we can design an experimental protocol to measure the quantum concurrence:

\begin{algorithm}[H]
\caption{Concurrence Measurement Protocol}
\label{concurrence measurement protocol}
\begin{algorithmic}[1]
\Require
Prepare $|\Psi\rangle$ sequentially for $N_U\times N_M$ times.
\Ensure
Probability distribution of measurement outcomes conditioned on evolution unitary $P(\vec{s}|U)$.
\For{$i= 1~\text{\textbf{to}}~N_U$} 
 \State Randomly pick a unitary matrix in every unitary ensembles to construct the evolution matrix $U=\bigotimes_{i=1}^nU_i$. 
 \State Operate $U$ on $\Psi$ to get $U|\Psi\rangle\langle\Psi|U^{\dagger}$. 
 \For{$j= 1~\text{\textbf{to}}~N_M$} 
  \State  Measure $U|\Psi\rangle\langle\Psi|U^{\dagger}$ in the computational basis $\{|\vec{s}\rangle\}$.
  \State Record the measurement results.
  \EndFor
 \State Estimate the probabilities $\mathrm{Pr}(\vec{s},U)$.
  
\EndFor
\State Do the data postprocessing given in Eq.~\eqref{eq:concurrence res} for $C_n(\Psi)$.
\end{algorithmic}
\end{algorithm}

 Then we have
\begin{equation}\label{eq:concurrence res}
\begin{aligned}
C_n(\Psi)=2\sqrt{1-3^n\sub{\mbb{E}}{U}P(\vec{s},U)^2}
\end{aligned}
\end{equation}

\begin{proof}
Substituting Born's rule and Eq.~\eqref{eq:con 2}, one can prove
\begin{equation}
\begin{aligned}
3^n\sub{\mbb{E}}{U}P(\vec{s},U)^2&=3^n\tr\left[(|\vec{s}\rangle\langle\vec{s}|)^{\otimes 2}\Phi^{2\otimes n}(\Psi^{\otimes 2})\right]\\
&=3^n\tr\left\{\left[\Phi^2(|s_1\rangle\langle s_1|^{\otimes 2})\otimes\cdots\otimes\Phi^2(|s_n\rangle\langle s_n|^{\otimes 2})\right]\Psi^{\otimes 2}\right\}\\
&=\tr\left[(P_+^1\otimes\cdots\otimes P_+^n)\Psi^{\otimes 2}\right]\\
&=\langle \Psi|\langle\Psi|P_+^1\otimes\cdots\otimes P_+^n|\Psi\rangle|\Psi\rangle,
\end{aligned}
\end{equation}
so that 
\begin{equation}
\begin{aligned}
2\sqrt{1-3^n\sub{\mbb{E}}{U}P(\vec{s},U)^2}&=2\sqrt{1-\langle \Psi|\langle\Psi|P_+^1\otimes\cdots\otimes P_+^n|\Psi\rangle|\Psi\rangle}\\&=\sqrt{\langle\Psi|\langle\Psi|A|\Psi\rangle|\Psi\rangle}\\&=C_n(\Psi).
\end{aligned}
\end{equation}

\end{proof}

\end{document}